\documentclass[reqno]{amsart}
\usepackage{amsmath,amssymb,latexsym,amscd}
\newcommand{\h}{\hat h}
\newcommand{\F}{\hat F}
\newcommand{\e}{\epsilon}
\newcommand{\te}{\tilde \epsilon}
\newcommand{\va}{\vartheta}
\newcommand{\tva}{\tilde{\vartheta}}
\newcommand{\tF}{\tilde F}
\newcommand{\tih}{\tilde h}
\newcommand{\tP}{\tilde \Pi}
\newcommand{\beq }{\begin{equation}}
\newcommand{\eeq }{\end{equation}}
\newtheorem{theorem}{Theorem}
\newtheorem{lemma}{Lemma}
\newtheorem{proposition}{Proposition}
\newtheorem{corollary}{Corollary}
\newtheorem{definition}{Definition}
\theoremstyle{definition}
\newtheorem{remark}{Remark}
\newtheorem{example}{Example}
\numberwithin{equation}{section}

\begin{document}
\title[Supplementary balance laws]%
{\textbf{Supplementary balance laws and the entropy principle.}}
\author{Serge Preston}(\address{Department of Mathematics and Statistics, Portland State University,
Portland, OR, U.S.}) \email{serge@mth.pdx.edu}

\begin{abstract}
In this work we study the mathematical aspects of the development in the continuum thermodynamics known as the "Entropy Principle". It started with the pioneering works of B.Coleman, W.Noll and I. Muller in 60th of XX cent. and got its further development mostly in the works of G. Boillat, I-Shis Liu and T.Ruggeri.  "Entropy Principle" combines in itself the structural requirement on the form of balance laws of the thermodynamical system (denote such system $(\mathcal{C})$) and on the entropy balance law with the convexity condition of the entropy density. First of these requirements has pure mathematical form defining so called "supplementary balance laws" (shortly SBL) associated with the original balance system. Vector space of SBL can be considered as a kind of natural "closure" of the original balance system.  This space includes the original balance laws, the entropy balance, the balance laws corresponding to the symmetries of the balance system and some other balance equations. We consider the case of Rational Extended Thermodynamics where densities, fluxes and sources of the balance equations do not depend on the derivatives of physical fields $y^i$. We present the basic structures of RET:  Lagrange-Liu equations,"main fields", and dual formulation of the balance system. We obtain and start studying the defining system of equations for the density $h^0$ of a supplementary balance law. This overdetermined linear system of PDE of second order determines all the densities $h^0$ and with them, due to the formalism of RET, the fluxes and sources of SBL. Solvability conditions of defining system delivers the constitutive restrictions on the balance equations of the original balance system. We illustrate our results by some simple examples of balance system and by describing all the supplementary balance laws and the constitutive restrictions for the Cattaneo heat propagation system. \par
\end{abstract}
 \maketitle
\today
 \tableofcontents

\section{\textbf{Introduction.}}
With any system $\mathcal{P}$ of partial differential equations for the fields $\{y^i (x^\mu),\ i=1,\ldots, m\}$ defined in a domain $D \subset X$ of the base (typically physical or material space-time) manifold $X$ there are associated several classes of objects carrying important information on this system and its solutions. \par
First class consists of symmetry groups of different type and levels of generality (see \cite{O}) and the corresponding Lie algebras of infinitesimal symmetries. Principal characteristic property of these symmetry groups is that the transformations from these groups acts on the solutions of system $\mathcal{P}$ transforming them into another solutions.\par  Second class consists of the conservation or more generally, balance laws of different types satisfied by all the solutions of system $\mathcal{P}$.\par
 Relations between these two classes of objects associated with the system $P$ are well know in the Variational Calculus.  They are are formulated as the theorems of Noether type of different level of generality (\cite{O}).  The most well known for the Lagrangian formalism is the original Noether theorem with the complement of Bessel-Hagen (see discussion in \cite{O})), the most general result in this case is the result obtained by P.Olver (\cite{O1}) and, independently, by A.Vinogradov, \cite{Vin}). These results establish the bijective correspondence between the equivalence classes of the conservation laws for the Euler-Lagrange system of equations and the equivalence classes of variational symmetries of the action functional (see \cite{O}, Thm. 5.42).\par

On the other hand, in the Continuum Thermodynamics, the II law postulate the existence of an additional (to the basic balance system) balance law - \emph{entropy law} $h^{\mu}_{,\mu}=\Sigma$ such that the part of the source term $\Sigma$ corresponding to the internal entropy production in the system is nonnegative.\par
 In the mid-60th of XX cent., it was suggested (B. Coleman and W. Noll,\cite{CN}) that the following statements forming the "\textbf{Entropy Principle}" are true for any thermodynamical continuum system of RET type:
\begin{enumerate}
\item Entropy balance \beq (\rho s)_{,t}+S^{A}_{,x^A}=\sigma \eeq is satisfied by all the solutions of the basic balance system () of the thermodynamical system.
\item Entropy density $h^0 (x,y)$ is the concave function of fields $y^i$.
\end{enumerate}

Notice that the first statement above has sense for any system of balance equations.  Thus, one can introduce the space of the "\emph{supplementary balance laws}" for a given balance system $\bigstar$, see ().\par

In the terms of the balance law () the II law of thermodynamics requires (in the Coleman-Noll form, (see \cite{CN})) that in the absence of the outside sources of entropy, the Clausius-Duhem inequality
\beq
\frac{\partial \rho s}{\partial t}+\partial_{x^A}(\rho s v^i +\frac{q^i}{T})\geqq 0
\eeq
is true for all the solutions of the balance system $\bigstar$.\par
Later on, I.Muller suggested (\cite{Mu1}) that the nonconvective entropy flux $S$ may have more general form then $\frac{q^i}{T}$ and has to be defined by the requirement of compatibility with the constitutive relation of the system $\bigstar$ together with the entropy density $h^0$. this suggestion can be naturally rephrased as the fulfillment of the requirement (1) in the list above.\par

As a result of these developments, the entropy principle crystalized as the \textbf{Amendment} (see \cite{ME}) to the II law of thermodynamics having sense for continuum thermodynamical systems of higher order (i.e. where the fluxes $F^{\mu}_{i}$ and sources $\Pi_{i}$ may depend on the derivatives of fields $y^i$ up to some order). This amendment requires that \textbf{the constitutive relations of the basic balance system $\bigstar$ and the entropy inequality (i.e. expression of the entropy density $s$ and the entropy flux $S$ in terms of the basic fields and their derivatives) should be such that any solution of the balance system $\bigstar$ would at the same time satisfy to the entropy inequality (1.2).}\par

In most works it is silently assumed that the entries in the entropy balance - $s,S,Q$ are the functions of $x^\mu, y^i$ and the derivatives \textbf{to the same order as the entries in the system} $\bigstar$. This, in particular, guarantees the removal, from the considerations, of trivial balance laws of the  II order, see \cite{O}.\par
This leaves the question about the structure of the right side $\Sigma$ of the entropy balance.  In continuum mechanics and in the material science the form of the entropy production function is of the utmost importance (\cite{JCL,Mu}). The entropy "production" splits $\Sigma=\Sigma^{out}+\Sigma^{in}$  into internal (\emph{real production}) and external (outside source) terms. After this, the II law of thermodynamics crystallizes as the requirement that $\Sigma^{in}\geqq 0$.\par

This leads naturally to the problem of description of all the balance laws satisfying the first requirement above (\textbf{supplementary balance laws}) for a given balance system $\bigstar$ and modulo the trivial balance laws  and extraction those of them that satisfy to the positivity condition.\par
In Section 2 we introduce basic setting: base (space-time) space $X$, configurational bundle $\pi:Y\rightarrow X$ and the balance systems $\bigstar$ and collect the notations that are used in the paper. In section 3 we introduce supplementary balancer laws (SBL) for a given balance system, give examples of such SBL.  In Section 4 we describe structures of affine subbundles of 1-jet bundle $J^{1}(\pi)\rightarrow Y$ corresponding to a balance system $\bigstar$ and a supplementary balance law for this system.  Relation between these affine subbundles define the main fields $\lambda^i$ and the Lagrange-Liu system of equations for supplementary balance laws. In section 5 we present, following the monograph of I.Muller and T.Ruggeri, \cite{MR} a brief review of basic notions of Rational Extended Thermodynamics including transition to the main fields variables and the symmetrical hyperbolic form of the balance system. In section 6 we study the LL-system for a case where main fields are functionally independent and get the geometrical form of the Boillat-Miller-Ruggeri theorem on the bijection between the SBL with the functionally independent mean fields $\lambda^i$ and the density functions $h^0$ with the non-degenerate Hessian satisfying to the \emph{defining linear system} of the PDE of 2 degree. In Sections 7 we present some simple examples of solving defining system.  In Sec.8 we show that in the case of 2 scalar fields $y^1,y^2$ in $1+1$ space-time, defining system reduces to one 2 order linear PDE that has to be hyperbolic for the balance system to have SBL with the density $h^0$ whose Hessian is definite. In section 9 we study the defining system for $h^0$.  We show that this system is generically subholonomic, find conditions on the flux fields $F^{A}_{i}(y)$ when this system is elliptic or holonomic.   In section 12 we analyze the structure of supplementary balance laws for the Cattaneo heat propagation system. We describe all the supplementary balance laws, determine the form of internal energy for the case when there are SBL and its structure when the production of the SBL is positive - formulation of II law of thermodynamics in terms of the defining function $\lambda^{0}(\va)$ of the "entropy" SBL.

\section{\textbf{Settings and Notations.}}

\subsection{Space-time manifold $X$.}

A state of material body will be described by the collection of the time-dependent fields $\{ y^i
,i=1,\ldots ,m\}$ defined in a domain $B\subset E^n$ of the physical, material or ... space with the
boundary $\partial B$.\par   Product of closure $\bar{B}$ and the time axis $T$ - $X=T\times \bar B$ is the space-time manifold - the cylinder in the Newtonian or Lorentz space-time.  We assume that the Pseudo-Riemannian metric $G$ is defined in the space-time $X.$   An example of
such a metric is the Euclidian metric $g=dt^2+h$ in the Newtonian space-time (\cite{MH}) or Minkowsky metric in the space-time of special relativity. We introduce (global) coordinates
$x^\mu ,\mu =1,2,3$ in $B$ (possibly induced from the global coordinates in $E^3$) and the time variable $t=x^0$. \par

Denote by $\eta$ the volume n-form $\eta=\sqrt{\vert G\vert}dt\wedge dx^1 \wedge dx^2 \wedge dx^n$
corresponding to the metric $G$.
\par

\subsection{State (configurational) bundle $\pi:Y\rightarrow X$.}
Basic fields of a continuum thermodynamical theory $y^i$ take values in the space $U\subset \mathbb{R}^m$ which we will call the
 {\bf basic state space} of the system.

Following the framework of a classical field theory (see \cite{BSF,FF}) we organize these fields in the
bundle
\[
\pi_{U}: Y\rightarrow X,\ X=\mathbb{R}_{t}\times {\bar B}, \ Y=X\times U
\]
with the base $X$ and the fiber
$U$.
\par
Denote by $s=s(x)$ sections of the configurational bundle - collection of values of fields $x\rightarrow \{ y^i(x)),i=1,\ldots,m \},\ x\in D$, where $D\subset X$ is an open subset of the space-time manifold $X$.\par
We will call variables $y^i$ - the \emph{vertical variables}.  differential of a function $f\in C^{\infty}(Y)$ along the fibers of the configurational bundle will be called - \emph{the vertical differential}: $d_{v}f=\sum_{i}\frac{\partial f}{\partial y^i}dy^i.$\par

We will also use the notation $d_{\mu}$ for the total derivative by $x^\mu$:
\[
d_{\mu}f=\frac{\partial f}{\partial x^\mu}+y^{i}_{,x^\mu}\frac{\partial f}{\partial y^i}.
\]
Denote by $\pi_{10}:J^{1}\rightarrow Y$ the 1-jet bundle of the bundle $\pi:Y\rightarrow X$ (see \cite{BSF,S}). Local chart $(x^\mu ,y^i )$ in the domain $W\subset Y$ defines the local chart $(x^,\mu, y^i, z^{i}_{\mu})$ in the preimage $W^{1}=\pi^{-1}_{10}(W)\subset J^{1}(\pi)$. Coordinates $z^{i}_{\mu}$ in the fibers of 1-jet bundle over $Y$ are defined as the partial derivatives: $z^{i}_{\mu}(s(x))=\frac{\partial y^i}{\partial x^\mu}$ for a sections $s(x)=\{y^I (x)\}$.\par
A section $s(x):D\rightarrow Y$ with the components $y^i (x)$ determine the section $j^{1}s (x)$ of the bundle $J^{1}(\pi)\rightarrow X$ - 1-jet of the section $s(x)$ - by the formula
\[
j^{1}(x^\mu )=(x^{\mu}, y^i (x),z^{i}_{\mu}(x)=\frac{\partial y^i}{\partial x^\mu} (x).
\]
for a manifold $M$ we denote by $T(M)\rightarrow M$ - its tangent bundle, by $T^{*}(M)\rightarrow M$ - its cotangent bundle.  If $\pi:Y\rightarrow X$ is a bundle, we denote by $V(\pi)\rightarrow Y$ the vertical subbundle of the tangent bundle $T(Y)$ formed by the vectors tangent to the fibers of $\pi$: $V_{y}=\{\xi \in T_{y}(Y)\vert \pi_{*y}\xi =0$.  In a local chart $(x^\mu ,y^i)$ the fiber $V_{y}(\pi)$ has the basis of vector fields  $\frac{\partial}{\partial y^i}$.  Denote by $V^{*}(\pi)\rightarrow Y$ - vertical cotangent bundle - dual bundle to the vertical tangent bundle $V(\pi)$. In a local chart $(x^\mu ,y^i)$ the fiber $V^{*}_{y}(\pi)$ has the basis of 1-forms $dy^i,\ i=1,\ldots, m$.

\subsection{Balance Equations.}
Fields $y^i$ are to be determined as solutions of the field equations
for the currents $F^{\mu}_{i}$. Often (but not always) $F^{0}_{i}=y^i$ are the densities of fields $y^i$:
\begin{equation}
F^{ \mu}_{i,\mu}=F^{0}_{i,t}+F^{A}_{i,x^{A}}  =\Pi_{i},\ i=1,\ldots, n.
\end{equation}
Here $\Pi_{i}(y,x)$ is the {\bf production+source terms} of the component $y^i$ and $F^{
A}_{i}(y,x)\frac{\partial }{\partial x^A}$ - the {\bf flow } in the $i$-th equation. In RET theory these quantities
are assumed to be functions of the fields $y^i$ and, possibly, of the points $(t,x)=\{x^\mu \} \in X.$ Quite
one restricts attention to the case where there $F^{\mu}_{i} ,\Pi_i$ \emph{do not depend explicitly on the
space-time point} $x^\mu$.

\section{\textbf{Entropy principle and the supplementary balance laws.}}\par

   Let
\[
 (F^{\mu}_{i}(s(x)))_{,x^{\mu}}=\Pi_{i}(s(x)),\ i=1,\ldots, m,\hskip2cm (\bigstar )
\]
be a balance system of order zero for the sections $s(x)=\{y^i (x) \}$ with the densities $F^{0}_{i}(x,y)$, fluxes $F^{A}_{i}(x,y),A=1,2,3$ and the source/production terms $\Pi_{i}(x,y)$. This system can be written as the system on the 1-jet bundle $J^{1}(\pi)$ using the total derivatives:
\[
d_{\mu}F^{\mu}_{i}=\Pi_{i}.
\]
Equations of system $\bigstar$ are obtained from the last equating by taking pullback by the 1-jet $j^{1}s(x)$ of a section $s:x\rightarrow Y.$\par

 A natural question generalizing the "entropy principle" of Continuum Thermodynamics (\cite{MR}) is - are there, except of the linear combinations of balance equations of the system  $\bigstar$, nontrivial (see below) balance laws \emph{depending on the same variables} $x^\mu ,y^i$ \emph{that are satisfied by all the solutions of  the balance system} $\bigstar$.  Thus, we define

\begin{definition} Let the $\bigstar$ is a system of balance equations of the RET type.
\begin{enumerate}
\item
  We call a balance law
  \beq (K^{\mu}(s(x))_{,\mu}=Q(x,y(x)). \eeq with the density $K^{0}(x,y)$ and flux $K^{A}(x,y)$- the \textbf{supplementary balance law} for the system $\bigstar$ if \textbf{any} solution $ s(x)$ of the balance system ($\bigstar$) is at the same time solution of the balance law (3.1).
  \item A supplementary balance law (3.1) for a balance system $\bigstar$ is called one \textbf{of the entropy type} if for all solutions $s(x)$ of the balance system $\bigstar$ the entropy production term $Q(s(x))$ \textbf{is nonnegative}.
  \end{enumerate}
\end{definition}

\begin{example}
An interesting subspace of the space $\mathcal{SBL}_\mathcal{C} $ of supplementary balance laws, including the linear combinations of the balance laws of the system ($\bigstar$), is determined by the following

\begin{proposition} Let a vertical vector field $\xi =\xi^i \partial_{y^i}\in V(\pi)$ is such that
  the condition $F^{\mu}_{i}d_\mu \xi^{i}=0$ is fulfilled.
  Then the balance law
\[
s^*d(\xi^{i}F^{\mu}_{i}\eta_{\mu})=\xi^{i}\Pi_{i}\eta \Leftrightarrow
((\xi^{i}F^{\mu}_{i})(s(x)))_{,x^\mu}=(\xi^{i}\Pi_{i})(s(x))
\]
is the supplementary balance law for the balance system $\bigstar.$
\end{proposition}
\begin{proof} Follows from $s^* d(\xi^{i}F^{\mu}_{i}\eta_{\mu})=d(s^{*}\xi^{i}F^{\mu}_{i}\eta_{\mu})=
\xi^{i}(s)d(s^{*}(F^{\mu}_{i}\eta_{\mu})+s^{*}(F^{\mu}_{i}d_\mu \xi^{i})\eta =(\xi^i \Pi_{i})\circ s+0.$
\end{proof}
\end{example}

\begin{example} If a Lie group $G$ is the (geometrical) symmetry group of the balance system $\bigstar$ , it determines the family of the balance laws corresponding to the elements of Lie algebra $\mathfrak{g}$, see (\cite{Pr2,Pr3}).
\end{example}

\begin{example}
The \textbf{entropy principle} of Thermodynamics (see (\cite{ME,MR,Mu}, or Sec. below) requires that the entropy balance
\beq h^{\mu}_{,\mu}=\Sigma, \eeq
 with the entropy density $h^0$, entropy flux $h^A,A=1,2,3$, and entropy
production plus source $\Sigma=\Sigma_{s}+\Sigma_{p} $ of a given
theory. This requirement place a serious restrictions on the form of constitutive relation $\mathcal{C}$
and leads to the construction of a dual system in terms of Lagrange ("main") fields $\lambda $ considered below
(see Sec.3 or \cite{MR}).\par

In addition to the existence of the entropy balance
law, the second law of thermodynamics requires the fulfillment of the condition $\Sigma_{p} \geqq 0$ of \textbf{positivity of the entropy production}.
Thus, the "Entropy principle" is the \textbf{independent Amendment to the second balance law} (\cite{ME}) that places severe restrictions to the constitutive relations, more specifically to the form of entropy density, entropy flux and entropy production.\par
\end{example}
\begin{example}
System of Maxwell equations for the electrical field $\mathbf{E}$ and magnetic field $\mathbf{H}$ in the empty space ${\mathbb{R}}^3$
\beq
\begin{cases}
\mathbf{E}_{,t}-\nabla \times \mathbf{B}=0,\\
\mathbf{B}_{,t}+\nabla \times \mathbf{E}=0
\end{cases}
\eeq
is the balance system (of RET type). Conditions
\beq
\nabla \cdot \mathbf{E}=0, \ \nabla \cdot \mathbf{H}=0,
\eeq
admitted for the electrical and magnetic fields in the absence of electrical charges \emph{are supplementary balance laws for the system} (3.3).
\end{example}

\section{\textbf{Affine subbundles of the bundle $J^{1}(\pi)\rightarrow Y$,"main fields" $\lambda^i$ and LL-equations.}}

Consider a balance system ($\bigstar $)
\[
 (F^{\mu}_{i}\circ s)_{,{\mu}}=\Pi_{i}(s),\ i=1,\ldots, m,\hskip3cm \bigstar
\]
defined on the bundle  $\pi:Y\rightarrow X$,and calculate explicitly the derivatives in it.  We get
\beq
F^{\mu}_{i,x^{\mu}}(s)+F^{\mu}_{i,y^j}(s)s^{j}_{,\mu}= \Pi_{i}(s).
\eeq
Here $A=1,\ldots ,n;\mu=0,\ldots,n$. \par
 This can be written in the form
\beq
s^{*}\left[\sum F^{\mu}_{i,y^j}z^{j}_{\mu}\right] =\\
=s^{*}\left[ \Pi_{i}-F^{\mu}_{i,x^{\mu}}\right].
\eeq

We may consider these relations as defining, for each point $(x^\mu ,y^i)\in
Y$ \emph{through which a solution of the balance system $\bigstar$ passes}, the system of \emph{affine planes} $A_{i}(z),\ i=1,\ldots, m $ \emph{in the fibers} $J^{1}_{x,y}(\pi)$ of the bundle $J^{1}(\pi)$:

\beq \sum F^{\mu}_{i,y^j}(x,y)z^{j}_{\mu}=
\Pi_{i}(x,y)-F^{\mu}_{i,x^{\mu}}(x,y)\eeq

\vskip0.4cm

  For each $i$ affine planes $z\rightarrow A_{i}(x,y)$ form, upon varying of the point $(x,y)$, the affine
subbundle $Y\ni (x,y)\rightarrow A_{i}\subset J^{1}_{x,y}(\pi )$. Intersection $W$ of these planes generates the affine subbundle
$Y\ni (x,y)\rightarrow W(z)\subset J^{1}_{x,y}(\pi )$ (\emph{generically}) of codimension $m$.\par

Arguments presented above can be summarized in the following
\begin{proposition} A section $s\in \Gamma (D(s),Y)$ is the solution of the balance system $\bigstar$ if and only if the image of the domain $D(s)$ of section $s$ under the mapping $x\rightarrow j^{1}(x)$ in intersection with the fibers $J^{1}_{(x,y)}(\pi)$ of the bundle $J^{1}(\pi)\rightarrow Y$ lays in the affine subbundle $W_{(x,y)}$  of the fiber $J^{1}_{(x,y)}(\pi)$, i.e.
\[
j^{1}s(x)\in W_{s(x)},\ \forall x\in D(s).
\]
\end{proposition}

Previous considerations shows that
\begin{corollary}
 Let a balance system $\bigstar$ is locally solvable at each point $(x,y)\in Y$ (see Appendix I), then the affine subbundles $A_{i}$ and their intersection $W$ are correctly defined affine subbundles of the bundle $\pi_{10}: J^{1}(\pi)\rightarrow Y$ \emph{of the constant rank}.
\end{corollary}
\vskip0.4cm
\textbf{From now on we assume that the balance system $\bigstar$ is nondegenerate} in the sense of Appendix I.
\vskip0.4cm

 Let now a balance law (3.1) is defined in the same space $Y$ as the balance laws of the system $\bigstar $.  Then,
  it generated the affine subbundle of $J^{1}(\pi )$ with the fiber given by the affine hyperplane $A_{K}$

\beq
\sum_{\mu,j} K^{\mu}_{y^j}z^{j}_{\mu}=
Q-K^{\mu}_{,x^{\mu}}.
\eeq

Since for any point $z$ of a fiber $J^{1}_{x,y}(\pi )$ over a point $(x,y)\in Y$
there exists a section $s:X\rightarrow Y$ such that $s(x)=(x,y)$ and its 1-jet passes
through a given point $z$, the following algebraic formulation of the "entropy principle" is true.

\begin{proposition} Let the balance system $\bigstar$ is nondegenerate and let
\[ (K^{\mu}(s(x))_{,\mu}=Q(x,y). \]
is a balance law with the same domain $D^1 \subset Y$ as the balance system $\bigstar$.\par
The balance law (2.1) is the SBL for the system $\bigstar$ if and only if for any point $(x,y)\in
Y$ the intersection $W_{(x,y)}$ of the affine hyperplanes (4.3) is contained in the affine hyperplane $A_{K,(x,y)}$
(4.4).
\end{proposition}

In such a situation the following simple result of affine linear algebra is useful
\begin{lemma}
Let $E^N $ be a vector space (over a field $k$) and let $W$ be an affine subspace in $E$ equal to the
intersections of fibers of linear functionals $f_{i}:E\rightarrow k$: $W=\cap_{i}f^{-1}_{i}(c_{i}),\ c_{i}\in
k$.  Let $h:E\rightarrow k$ be a linear functional on $E$ and let $c\in k$.  Then the following statements
are equivalent
\begin{enumerate}
\item
$ W\subset h^{-1}(c), $
\item
$ h=\sum_{i}\lambda^{i}f_{i} $ for some $\lambda^{i}\in k$.
\end{enumerate}
\end{lemma}

Applying this lemma tot the affine hyperplanes $A_{i}$ and $A_{K}$ we get, the system of equations for a SBL (3.1) in the case of balance system of the RET type.
\begin{theorem} Let the balance system $\bigstar$ is nondegenerate system of the RET type and let $\sigma_{K}$ is a balance
law with the same domain $Y$ as the $\bigstar$.
The balance law (3.1) is the supplementary balance law for the balance system $\bigstar$ if and only if for any point
$\mathbf{z}\in Y$ there exist functions $\lambda^{i}(\mathbf{z})\in C^{\infty}(Y)$ (Lagrange
multipliers, main fields) such that

\beq \label{LL1}
\begin{cases}
\sum_{i}\lambda^{i}(z)
F^{\mu}_{i,y^j} =
K^{\mu}_{,y^j},\\
\sum_{i}\lambda^{i}(z)\left(\Pi_{i}-F^{\mu}_{i,x^{\mu}} \right) =Q(z)-K^{\mu}_{,x^{\mu}}.
\end{cases}
\eeq
\end{theorem}
For shortness we will call system of equations (4.5) - the LL-system for the balance system $\bigstar$.\par
Using the differential by variables $y^i$ (vertical differential) we can write the system of equations $4.5$ in
the invariant form
\beq \label{int} d_{v}K^{\mu}=\lambda^{i}(z)d_{v}F^{\mu}_{i}. \eeq

First subsystem of LL-system (4.5) determines $K^\mu$ up to addition of arbitrary functions of $x^\mu$. It is clear that together with $K^\mu$ all the functions of the form $\tilde K^\mu =K^\mu +{\bar K}^{\mu}(x)$ solve the first subsystem (4.5) \emph{with the same functions} $\lambda^i$. Second subsystem defined the source term $Q$ for the balance equation with the density-flux $K^\mu$ and the source term $Q+{\bar K}^{\mu}_{,\mu}$ for the balance law with the density-flux $\tilde K^\mu $. Difference of these two balance laws will have the form $d_{\mu}{\bar K}^{\mu}=K^{\mu}_{,\mu}$ i.e. is the trivial balance law of the second type independent on the fields $y^i$, i.e. of order $-1$.  This prove the following
\begin{lemma} For a fixed set of main fields $\lambda^i$ the system (4.5) determines the supplementary balance law $\sigma =\bar{K}^\mu\eta_{\mu}+\bar{K}^{\mu}_{,\mu}$ uniquely, up to a trivial balance law of second type and order $-1$.
\end{lemma}

\begin{example} Consider a special case where all the main fields vanish: $\lambda^i=0,\ i=1,\ldots, m$.  In this case equations (4.5) becomes statements that $K^{\mu }$ does not depend on $y^j$: $K^{\mu}_{,y^j}=0,\ \mu=0,1,2,3;\ j=1,\ldots,n$. In addition to this the source/production term takes the form
\beq
Q=K^{\mu}_{,x^{\mu}}.
\eeq
Formulating now the balance law corresponding to these $K^\mu ,Q$ we immediately see that $Q=d_{\mu}K^{\mu}$, i.e. that this balance law is the \textbf{trivial one of the second kind} of very special type (\cite{O}).
\end{example}

\begin{example} For the Cattaneo heat propagation balance system (see (12.1) below) with dynamical variables $\vartheta$ -temperature and $\mathbf{q}$ - heat flux the LL system (4.5) has the form
\beq
\begin{cases} \lambda^0 \e_{,\tva}+\tau_{,\tva}\lambda^A q^A  =K^0_{,\tva},\\ \ \lambda^0 \e_{,q^A}+\lambda^A \tau  =K^{0}_{,q^A}\end{cases};\
\begin{cases}    K^A_{,\tva}=\lambda^A,\\ \ K^{A}_{,q^B}=\lambda^0 \delta^{A}_{B}\end{cases},\ A,B=1,2,3 ,
\eeq
where instead of the temperature, the variable $\tva =\Lambda (\va )$ was used.
\end{example}

\begin{remark} Correspondence between the supplementary balance laws (3.1) and the sets $\{ \lambda^{i}(z)
\}$ of main fields is linear. This correspondence defined the mapping from the space of all
supplementary balance laws of the balance system $\bigstar$ to the space of sections $\lambda^{i}\partial_{y^i}$  of the \emph{vertical tangent bundle} $V(\pi)\rightarrow Y$. Vector fields with $\lambda^{i}$ \emph{constant} ($\nu$-horizontal with respect to a connection $\nu$, see \cite{Pr3}) corresponds to the linear combinations of the balance laws of the system $\bigstar$.
\end{remark}

\section{\textbf{ Formalism of Rational Extended Thermodynamics (RET)}.}

Here we describe, in a short form, the basic structure of the Rational Extended Thermodynamics developed by
I.Muller and T.Ruggeri, \cite{M,MR}.  For more complete presentation of the formalism of Rational Extended
Thermodynamics we refer to the monograph \cite{MR}, Chapter 3. Here we introduce only necessary material
in the form suited for our purposes.   \par

We will be using the notations and setting described in Sec.2. In this section we take $n=4$.\par
Balance equations for the fields $y^i(t,x)$ have the form $\bigstar$
\[
F^{ \mu}_{i,\mu}=F^{0}_{,t}+F^{A}_{i,x^{A}}  =\Pi_{i},\ i=1,\ldots, m.\hskip2cm \ bigstar
\]

\begin{remark}
In the rational Extended Thermodynamics one considers a case where balance equations are written \emph{for
all the basic fields in the state space and only for them} and where flows $F^{\mu}_{i}$ and productions
$\Pi_i$ depend on the fields $y^i$ \emph{but not on their gradients or time derivatives}. If the constitutive relations depend on the derivatives of fields $y^i$, the basic state space is extended by adding these derivatives to the list of basic variables $y^i$. \textbf{As a result, the balance system above is the system of quasi-linear PDE of the first order}.
\end{remark}

To close system of equations $\bigstar$ for the fields $y^i (x)$ one has to to choose the {\bf constitutive equations} of the body - to specify the densities $F^{0}_{i}$, flows $F^{A}_{i}$ and production terms $\Pi_{i}$ as the functions of $x^\mu , y^i$.  \par

\subsection{II law of thermodynamics and the Entropy Principle.}
In addition to the choice of constitutive relations for closing the balance system,  one has to satisfy the \textbf{II law of thermodynamics}.  This law requires the fulfillment of one additional balance equation - entropy balance, with additional condition of nonnegativity of internal entropy production (see below). Entropy density $h^0$ is not considered as an independent physical field.  In homogeneous thermodynamics constitutive law determine entropy as the function of energy and other extensive variables (see \cite{CA}). In continuum thermodynamics entropy density and flux are supposed to be functions of basic dynamical fields $y^i$ and their derivatives. \par
To facilitate the choice of constitutive relations and ensure the coherent physical behavior of the balance system of continuum thermodynamics, the \textbf{entropy principle} was suggested.\par
As we will see below, utilizing of the entropy principle facilitates the choice of constitutive relations for the original balance system plus the entropy principle and, in the case of RET allows to reduce this
process to the choice of {\bf entropy flow} 3-form and to the choice of production 4-forms subject to the
positivity condition.\par

We start with the formulation of the entropy balance.\par
\textbf{Entropy density} (per unit of volume) $h^0$ and the entropy flux $h^A,A=1,2,3$ are assumed (in RET) to be functions of the basic state variables $y^i$. These quantities satisfy to the balance law
\begin{equation}
d(h^{\mu}\eta_{\mu})=d_{t}h^{0}+d_{x^A}h^{A}=\Sigma^{ex}+\Sigma^{in} ,
\end{equation}
where the right side of entropy balance is the sum of \textbf{external entropy source} $\Sigma^{ex}$  and the \textbf{internal entropy production} $\Sigma^{in}$.\par  The II law of Thermodynamics is the requirement that the \textbf{internal production 4-form $\Sigma^{in}$ in the entropy balance is nonnegative}
\beq \Sigma^{in} \geqq 0. \eeq

{\bf Entropy principle} (introduced by B.Coleman, W.Noll  and clarified by I.Muller, \cite{Mu1}) consists of two parts:
\begin{enumerate}
\item Constitutive relations in the balance system $\bigstar$ and the components ($h^\mu, \Sigma$) of the entropy balance should be such that \textbf{any solution of the balance system $\bigstar$ would also satisfy to the entropy balance} (5.1).
\item (\emph{Convexity condition}). Entropy density $h^0(y)$ is a concave function of variables $y^i$:
\beq \frac{\partial^{2}h^0}{\partial y^i \partial y^j} \sim \text{negative\ definite} \eeq
\end{enumerate}

\begin{remark} Conceptually, the entropy principle formulated above (see also \cite{R2}) is the \textbf{mathematical Amendment} to the \textbf{physical II law of thermodynamics} (as it was explicitly formulated and supported by W.Muschik and H.Ehrentraut in \cite{ME}). It is a very powerful tool to formulate the entropy balance for a given balance system and to impose limitations to the constitutive relations of the original balance system.\par
On the other hand, as the reader will see in the next Chapter, this requirement is actually determine a natural extension of the original balance system to include all the balance equations (or, more physically reasonable, classes of equivalence of such balance equations modulo trivial ones) compatible with the original system. \par
It follows from this that the range of validity of the first part of the entropy principle is much larger then just the entropy balance.
\end{remark}
\begin{remark} Fulfilment of the convexity condition guaranteed the thermodynamical stability (see (\cite{MR}) or the Coleman-Noll Theorem in \cite{TN}, Sec.89)
\end{remark}

\begin{remark}$^*$
If the convexity requirement is satisfied, the symmetrical bilinear form \beq g_{ij}(y)=-
\frac{\partial^{2}h^0}{\partial y^i \partial y^j}
 \eeq
 can be considered as a \textbf{Riemannian metric} in the basic state
 space $U$.  This is the \emph{Wundlender-Ruppeiner thermodynamical metric}
 (\cite{M}).  It would be interesting to interpret the curvature of
 this metric in the context of RET.\par
\end{remark}

\subsection{{Convexity of entropy density and the main fields.}}

Requirement that the entropy balance equation (5.1) is fulfilled \textbf{for all solutions} $y=\{ y^i\}$ \textbf{of balance
system} $\bigstar$ leads to strong limitations on the form of constitutive equations. In order to apply the scheme below to more general systems we need to replace the fields $y^i$ with the new variables $w^i=F^{0}_{i},\ i=1,\ldots, m$.  This is possible only is the change of variables $y^i\rightarrow w^k$ is invertible.  Thus, we introduce the class of RET systems satisfying to this condition

\begin{definition} Balance system of RET type (of order zero) is called \textbf{regular} in a domain $U\subset Y$ if the  vertical \emph{differentials} $d_{v}F^{0}_{i}(y)\vert_{V(\pi)}$ \emph{form the coframe of the vertical tangent space} $V(\pi)$ at each point $y\in Y$.
\end{definition}

We restrict our considerations to the case of regular RET balance systems.  Regularity condition ensures that the \emph{densities} $w^i=F^{0}_{i}$ \emph{ of the balance laws $\bigstar$ are (locally) functionally independent as functions of $y^j$ }.  As a result the
functions $w^{i}$ can be taken as local coordinates in the fibers of the bundle $\pi$. This change of variables induces the change of frame, coframe and the corresponding change in the vertical differential $d_{v}=\partial_{y^i}\otimes dy^i\rightarrow \partial_{w^j}\otimes dw^j.$  We will keep the same notation for vertical differential - it is clear from the context which variables are using.\par

Notice that the conditions of the Entropy Principle, including the  convexity condition of the entropy density (- positivity of the minus Hessian of the function $h^0$) are invariant with respect to the change of variables $y\rightarrow w$.

\begin{remark} Introduced \textbf{regularity condition} guarantees that the balance system can be written in the normal form
\[
\partial_{t}\begin{pmatrix} y^1 \\ \ldots \\ y^m
\end{pmatrix} +\partial_{X^A}\begin{pmatrix} F^{A}_{1},\\ \ldots \\ F^{A}_{m}\end{pmatrix} =\begin{pmatrix} \Pi_{1},\\ \ldots \\ \Pi_{m}
\end{pmatrix},
\]
and, therefore, is nondegenerate as the evolutional system. If differentials $dF^{0}_{i},\ i=1,\ldots,m$ are \emph{functionally dependent on the fibers} of the bundle $\pi$, system describes evolution of some but not all fields $w^i$ and leaves some other fields or combination of fields to satisfy some stationary equation.
\end{remark}

The first condition of the Entropy Principle, formulated above, is equivalent to the following two statements:\newline
1.  There exist  functions $\lambda^{i}(w)$ (called Lagrange multipliers or "\textbf{main fields}") on the space $Y$ such that for all values of variables $w^i$
\begin{equation}
\frac{\partial h^\mu }{\partial w^i}=\lambda^{j}\cdot \frac{\partial F^{\mu }_{j}}{\partial
w^i}\Leftrightarrow dh^{\mu}=\lambda^{j}\cdot dF^{\mu}_{j},
\end{equation}
and
\begin{equation}
2.\ \Sigma =\lambda^{i}\Pi_{i}\geqq 0.
\end{equation}
For the proof of these statements we refer to Section 4 or \cite{MR}, Ch.3.\par
\par First of the equation (5.5) defines the Lagrange multipliers

\beq \label{change} \lambda^i=\frac{\partial h^0}{\partial w^i}. \eeq
Introduce the m-dim manifold $\Lambda^m$ with the coordinates $\lambda^{i}, i=1\ldots, m$ and the trivial (\emph{dual configurational}) bundle $\pi^{d}: X\times \Lambda \rightarrow X$.  Then the correspondence \ref{change} can be considered as the gradient mapping from the fibers of the bundle $\pi:Y\rightarrow X$ to the fibers of the dual bundle $\pi^d$. To define this mapping in an invariant way identify the manifold $\Lambda$ with the fiber(s) of the dual vertical bundle $V^{*}(\pi)\rightarrow Y \rightarrow X$ . This mapping is then the section of the bundle $V^{*}(\pi)\rightarrow Y$ defined by the function $h^0$.\par
Differentiating last relation by $w^j$ we get
\[
Hess(h^0 )=\frac{\partial^2 h^0}{\partial w^i \partial w^j}=\frac{\partial \lambda^i}{\partial w^j}.
\]
To qualify the invertibility properties of introduced gradient mapping $y\rightarrow \lambda$ we use the following result (\cite{GH},Ch.7, Sec.1.1.
\begin{lemma} Let $D\subset \mathbb{R}^N$ be an open domain in $R^N$ with the coordinates $x^i$ and $f\in C^{k}(D),k\geqq 2$ is a smooth function in $D$.  The gradient mapping $\phi:D\rightarrow \mathbb{R}^N$: $\phi^{i}(x)=\frac{\partial f}{\partial x^i}$.  Then
\begin{enumerate}
\item The gradient mapping $\phi$ is locally invertible if
\beq \label{one}
det(Hess(f))\ne 0\ in\  the\  domain\  D.
\eeq
\item If the domain $D$ is convex and the Hessian matrix $Hess(f)=\left(f_{,x^ix^j} \right)$ is positive (or negative) definite on $D$, then the gradient mapping $\phi$ is $C^{k-1}$-diffeomorphism of $D$ onto $D^*=\phi(D)$.
\end{enumerate}
\end{lemma}
\begin{proof} If the property \ref{one} holds, then, due to the Inverse Function Theorem, locally, $\phi$ is a diffeomorphism of the class $C^{k-1}$.\par
Let now $D$ is convex and the Hessian of the function $f$ is positive definite at all points of the domain $D$. We have to prove that $\phi$ is one to one in $D$. Let $\phi(x_1)=\phi(x_2)$ for some $x_1,x_2 \in D$ and let $x=x_1 -x_2$. since $D$ is convex, points $x_1 +tx,\ 0\leqq t\leqq 1$ belong to $D$. Then the correspondence $t\rightarrow A(t)=Hess(f)(x_1 +tx)$ for $t\in [0,1]$ defines a continuous matrix valued function with the positive definite $A(t)$. We have, using euclidian scalar product in $\mathbb{R}^N$,
\beq 0=<x,\phi(x_1 )-\phi(x_{2} )>= \langle x, \int_{0}^{1}\frac{d}{dt}\phi(x_1 +tx)dt\rangle = \int_{0}^{1} ,x,A(t)x>dt.
\eeq
Positivity of matrix $A(t)$ for all $t$ shows that $x=0$, i.e. $x_1 =x_2$.
\end{proof}

\begin{corollary}\textbf{If the entropy density $h^0$ is a strongly convex (or concave) function of its arguments
$w^i$, then the change of variables \ref{change} $\{ w^i\}\rightarrow \{ \lambda^j\} $ is {\bf globally invertible}} and defines the
diffeomorphic mapping
\beq \wp : U\rightarrow \Lambda \eeq
 from the state space $U$ onto the space $\Lambda= \wp (U) \subset R^m$ of values of variables $\lambda =\{ \lambda^{i} \}.$
\end{corollary}

\subsection{Dual formulation.}
Using introduced change of variables, one may present all the quantities as the functions of \textbf{dual (Lagrange) variables} $\lambda^{i}$:
\begin{equation}
\tF^{ \mu}_{i}=\tF^{ \mu}_{i}(\lambda ), \tP_{i} \equiv \tP_{i}(\lambda ); \tih^\mu =\tih^\mu (\lambda).
\end{equation}
Introduce the {\bf four-vector potential} (3-form in 4-dim space-time)
\begin{equation}
\h=\h^\mu \eta_{\mu}=[\lambda^i \cdot \tF^{ \mu}_{i}-\tih^\mu (\lambda )]\eta_{\mu}=\lambda_{i}\tF_{i}-\tih.,
\end{equation}
where last equality is the definition of 3-forms $\tF_{i}.\tih$.\par
In terms of functions $\hat h^\mu$ the relation (5.5) takes the simple form
\begin{equation}
d\h^\mu (\lambda) =\F^{ \mu}_{i}d\lambda^i,
\end{equation}
here summation by repeating indices is assumed.\par
From the relation (5.12) it follows that
\begin{equation}
\tF^{\mu}_{i}=\frac{\partial \h^\mu}{\partial \lambda^i} \Rightarrow \tih^{\mu}(\lambda
)=-\h^{\mu}+\lambda^{i}\cdot \frac{\partial \h^\mu}{\partial \lambda^i}.
\end{equation}

As a result, $4n+4$ constitutive functions $\tF^{\mu}_{i}$ and $\tih^{\mu}(\lambda )$ can, in terms of
$\lambda $ variables be derived from the 4 functions $\h^\mu $ - coefficients of 3-form $\h$.  This representation is dual to the relations  $d_{v}h^\mu =\lambda^i d_{v}F^{\mu}_{i}$ (see (5.5)).\par
Taking derivatives in the first equation (5.14) we get $\frac{\partial \tF^{\mu}_{i}}{\partial \lambda^j}=\frac{\partial^2 \h^\mu}{\partial \lambda^i
\partial \lambda_j}.$\par

Taking in (5.12) $\mu=0$ we see that the function ${\hat h}^{0} (\lambda)$ is the Legendre transformation (\cite{GF}) of the function $h^0 (y)$:
\beq
{\hat h}^{0} (\lambda)=\frac{\partial h^{0}}{\partial w^i}w^i -h^{0}(w^k)=\lambda^{i}w^i -h^{0}(w^k),
\eeq
where $w^i$ is considered in the last term as the functions of $\lambda^k$ obtained by solving equations (5.7). \par
\begin{lemma} If the Hessian $Hess(h^0 )$ of the function $h^0 (w)$ is negatively definite, so is the Hessian $Hess ({\hat h}^{0})$ of the function ${\hat h}^{0}(\lambda)$.
\end{lemma}
\begin{proof} From $d{\hat h}^{0}=\lambda^i dw^i +y^i d\lambda^i -h^{0}_{,w^k}dw^k =w^i d\lambda^i,$ we see that $w^j=\frac{\partial {\hat h}^{0}}{\partial \lambda^j}$.  Using this we get
\[
\frac{\partial^{2}{\hat h}^{0}}{\partial \lambda^i \partial \lambda^j}=\frac{\partial w^i}{\partial \lambda^j}=(\frac{\partial \lambda^{i}}{\partial w^j})^{-1}=(\frac{\partial^{2}{ h}^{0}}{\partial w^i \partial w^j})^{-1}<0.
\]
\end{proof}

After presenting currents $\tF^{\mu}_{i}$ in the form (5.11) what is left of the requirements of entropy
principle (provided the condition of convexity of $h^0$ is fulfilled) is the {\bf residual inequality}
\begin{equation}
\Sigma(\lambda )=\lambda^i \Pi_i \geqq 0.
\end{equation}
Two statements containing here determine the entropy production $\Sigma $ in terms of the production
4-forms $\Pi_i $ and {\bf require} non-negativity of $\Sigma $.\par

Reversing the arguments leading to the statements (5.14) and (5.16) one proves the following basic
result of RET leading to the dual formulation of system $\bigstar$ of balance equations and the entropy principle (5.1)
\begin{theorem}{\cite{MR}}
The following statements are equivalent under the condition of the convexity of entropy density
$h^{0}(w^i)$ as the function of fields $w^i$:
\begin{enumerate}
\item  Entropy principle is fulfilled for the balance equations
$\bigstar$ and the entropy balance equation (5.1) for given constitutive functions
$F^{\mu}_{i}(y),\Pi_{i}(y),h^{\mu}(y),\Sigma(y)$.

\item Constitutive fields $F^{\mu}_{i}(y),h^{\mu}(y),\Sigma(y)$ are obtained by
the relations (5.7),(5.14), (5.15) from the four-potential $\h(\lambda )$ (formal 3-form) and the
production 4-forms $\Pi_i (\lambda)\eta$ for which the residual inequality
\[
\lambda^i \Pi_i \geqq 0
\]
is fulfilled.
\end{enumerate}
\end{theorem}

\vskip1cm
\emph{2. DEFINE in Appendix the symmetric hyperbolic systems.}
\vskip1cm

Combining balance equations $\bigstar$ with the change of variables $w^i\rightarrow \lambda^j$ and using the relation 5.14, we rewrite the balance system $\bigstar$ in the form
\begin{equation}
\frac{\partial \tF^{ \mu}_{i}}{\partial \lambda^{j}}\cdot \frac{\partial \lambda^{j}}{\partial
x^{\mu}}=\tP_{i}(\lambda ) \Leftrightarrow \frac{\partial^2 \h^\mu}{\partial \lambda^i
\partial \lambda^j}\frac{\partial \lambda^{j}}{\partial
x^{\mu}}=\tP_{i}(\lambda ),\ i=1,\ldots ,m.
\end{equation}
 This is the system of first order PDE for the fields $\lambda^{i}(w)$ (of the type first suggested by S.Godunov, see \cite{Gd,Gd2}) with the symmetric matrix $\frac{\partial^2 \h^\mu}{\partial \lambda^i
\partial \lambda^j}.$  Lemma 4 proves that the matrix $\frac{\partial^2 \h^0}{\partial \lambda^i \partial \lambda^j}$ is \emph{negatively definite}. This finishes the proof of the following lemma where symmetry of the matrices $\frac{\partial^2 \h^\mu}{\partial \lambda^i \partial \lambda^j}$ is obvious.
\begin{lemma} Let the function $h^0$ is concave or convex and let $w^i \rightarrow \lambda^j$ is the change of variables $w^i$ to the main fields $\lambda^i$.  In the variables $\lambda^i$, the RET balance system $\bigstar$ has the form of \emph{symmetric hyperbolic system} (by Friedrichs, see \cite{FL})
\beq
\frac{\partial^2 \h^\mu}{\partial \lambda^i
\partial \lambda^j}\frac{\partial \lambda^{j}}{\partial
x^{\mu}}=\tP_{i}(\lambda ),\ i=1,\ldots ,m
\eeq
for the Lagrange fields $\lambda^i$.
\end{lemma}

\section{\textbf{Main fields: Case of functionally independent} $\mathbf{\lambda^i}$.}

 In the RET case the Lagrange-Liu system of equations (4.5) for the fluxes and the source of the SBL of the balance system $\bigstar$ has the form (we are using variables $w^i$ considering system $\bigstar$ regular (see Definition 2)).
\beq
\begin{cases}
\sum_{i}\lambda^{i}(w)d_{v}F^{\mu}_{i} &=d_{v}K^{\mu},\ \Leftrightarrow
\sum_{i}\lambda^{i}(w)F^{\mu}_{i,w^j}=K^{\mu}_{,w^j},\ \forall \mu=1,\ldots, n;j=1,\ldots ,m, \\
\sum_{i}\lambda^{i}(w)\Pi_{i}(w) &=Q(w)-K^{\mu}_{,x^\mu},
\end{cases}
\eeq
with $d_{v}$ being in this case the vertical differential in the bundle $\pi$:
\[ d_{v}f(x,w)=f_{,w^i}dw^i.\]\par

In this and the next section we omit the sign of variables $x^\mu$ in the arguments of functions
silently assuming that this dependence is present.  No derivatives by $x^\mu$ will be calculated in these
sections, so this will not lead to any ambiguities.\par
 Fix an index $\mu$ in the first of these relations and notice that locally in $Y$ the
 exactness of the form in the left side is equivalent to its vertical closeness i.e. to the following integrability
 condition

\beq d_{v}(\sum_{i}\lambda^{i}(w)d_{v}F^{\mu}_{i})=\sum_{i}d_{v}\lambda^{i}(w)\wedge
d_{v}F^{\mu}_{i}(w)=0. \eeq

Here we restrict our considerations to the case of regular RET balance systems (see Definition 2 above).  Regularity condition ensures that the
functions $w^{i}=F^{0}_{i}(y)$ can be taken as local coordinates in the fibers of the bundle $\pi$. We will keep the same notation for vertical differential with respect to variables $w^i$- it is clear from the context which variables we are using.\par

As the next step, we apply Cartan Lemma, see Appendix 3, to the second form of the condition (6.2) for $\mu =0$  written in terms of the vertical variables $w^i$:
\[
\sum_{i}d_{v}\lambda^{i}(w)\wedge dw^i=0.
\]
We get: $d_{v}\lambda^{i}=\sum_{j}A_{ij}dw^j,\ A_{ij}=A_{ji} \rightarrow \lambda^{i}_{,w^j}=A_{ij}=A_{ji}=\lambda^{j}_{,w^i}$
for the functions $A_{ij}(w)$. Applying the mixed derivative test we finish the proof of the following
\begin{lemma}  If the RET constitutive relation $C$ is regular (functions $F^{0}_{i}$ are vertically functionally independent) then the condition (6.2) for $\mu=0$ is fulfilled if and only if locally (and in a domain $W\subset Y$ with star-shaped fibers over $X$, globally) there exists a function $h^{0}\in C^{1}(W)$ such that
\[
\lambda^{i}=\frac{\partial h^{0}}{\partial w^{i}}.
\]
\end{lemma}
Applying this lemma to the condition (6.2) we find that Lagrange multipliers $\lambda^{i}(x,y)$ have
the form

\[
\lambda^{i}(x,y)=\frac{\partial h^{0}}{\partial w^{i}}\vert_{w^{i}=F^{0}_{i}(y)},
\]
for a smooth function $h^{0}\in C^{1}(W).$ Substituting these expressions into (6.1) we find that one
should have

\beq
\begin{cases}
K^{0}&=h^0 +{\bar K}^{0}(x),\\
 d_{v}K^{A}(y)&=\sum_{i}\frac{\partial h^{0}}{\partial w^{i}}\vert_{w^{j}=F^{0}_{j}(y)}d_{v}F^{A}_{i},\ \forall A=2,\ldots,n,\\
 Q(y)&=\sum_{i}\frac{\partial h^{0}}{\partial w^{i}}\vert_{w^{j}=F^{0}_{j}(y)}\Pi_{i}(y)-K^{\mu}_{x^\mu}.
\end{cases}
\eeq

These relations defines uniquely the production term $Q(y)$ in (3.1). Flux terms $K^{A},A=2,\ldots,n$ in
the balance law (3.1) can be found up to the adding of an arbitrary function of $x$ \emph{provided that for all $A$ the
1-forms on the right are closed with respect to the vertical differential} $d_{v}$. Closeness condition(s) form the set of necessary and locally sufficient conditions for a function $h^0$ to generate a secondary balance law of the system $\bigstar$.  Specifically, these solvability conditions have, in terms of variables $w^{i}=F^{0}_{i}(y)$ the form  (here $d_{v}\tilde{F}^{0}_{j}=dw^{j}$ and $\tilde F^{A}_{i}(w)=F^{A}_{i}(y(w))$)

\beq d_{v}\frac{\partial h^{0}}{\partial w^{i}}(w^k =\tilde{F}^{0}_{k})\wedge d_{v}\tilde{F}^{A}_{i}=0 \Leftrightarrow
\frac{\partial^2 h^{0}}{\partial w^{i}\partial w^{j}}(\tilde{F}^{0}_{k})dw^j\wedge d_{v}\tilde{F}^{A}_{i}=0,\forall A=2,\ldots,n.
\eeq
Here ${\tilde F}^{A}_{i}(w)=F^{A}_{i}(y^k =y^k(w))$.\par

As a result we have proved the following

\begin{theorem} Let ($\bigstar$) be a balance system of the RET type with the functionally independent by vertical variables $y^i$ density functions $F^{0}_{i}$.\par
  There is a bijection between the \emph{ regular} supplementary balance laws (3.1) of the  balance
 system ($\bigstar$) defined in an open subset $W\subset Y$ such that the multipliers $\lambda^{i}=\frac{\partial K^0}{\partial w^i}$ are
\textbf{functionally independent} and the smooth functions (potentials) $h^{0}(w^{1},\ldots ,w^{m})\in
C^{\infty}(W)$, modulo addition of an arbitrary function $\pi^{*}f(x),f\in C^{\infty}(\pi(W))$, with the
nondegenerate \textbf{vertical} Hessian $Hess(h^0)=\frac{\partial^{2}h^{0}}{\partial w^i
\partial w^j}$ such that
\[
\frac{\partial^2 h^{0}}{\partial w^{i}\partial w^{j}}(w^{k}=F^{0}_{k})d_{v}\tilde{F}^{0}_{j}(w)\wedge
d_{v}\tilde{F}^{A}_{i}(w)=0,\ A=2,\ldots ,n.
\]

This bijection is given by the relations
 \beq
\begin{cases}
K^{0}(y)&=h^{0}(F^{0}_{k})(y),\\
d_{v}K^{A}(y)&=\sum_{i}\lambda^{i}d_{v}F^{A}_{i},\ \forall A=2,\ldots,n,\\
Q(y)&=\sum_{i}\lambda^{i}\Pi_{i}(y)+K^{\mu}_{,x^\mu},
\end{cases}
\eeq where Lagrange multipliers $\lambda^{i}$ are defined by \beq \lambda^{i}(y)=\frac{\partial
h^0}{\partial w^{i}}\vert_{w^{k}=F^{0}_{k}(y)}. \eeq
\end{theorem}
\begin{example}
\begin{enumerate}
\item It is clear that all the linear functions $h^{0}(w)=a^{i}w_i+b$ satisfy to these conditions.  They
correspond to the linear combinations of balance laws of the system ($\bigstar$) with constant coefficients.\par
\item Let now
$h^{0}(w)=\frac{1}{2}k_{ij}w^{i}w^{j}$ be a homogeneous quadratic function of its arguments ($k_{ij}\in \mathbb{R}$ - const).
Then, the condition of Theorem 3 takes the form
\[
k_{ij}dw^{j}\wedge d_{v}{\tilde F}^{A}_{i}(w)=0 \Leftrightarrow \ (k_{ij}{\tilde F}^{A}_{i,l})dw^j \wedge dw^l \Leftrightarrow k_{ij}{\tilde F}^{A}_{i,l}-k_{il}{\tilde F}^{A}_{i,j}=0\ \forall A,j,l.
\]
In terms of the quantities $D^{A}_{j}=k_{ij}{\tilde F}^{A}_{i}$, last condition takes the form $D^{A}_{j,l}=D^{A}_{l,j}$. This condition can be rewritten in the form of 1-form closeness
\beq
d_{v}(D^{A}_{j}dw^j )=0 \Rightarrow ^{locally}\ D^{A}_{j}=\frac{\partial q^A}{\partial w^j},
\eeq
for some functions $q^A (w)$. Returning to the functions ${\tilde F}^{A}_{i}$ we get similar \emph{constitutive} condition for existence of a SBL with the quadratic by $w^i$ density $h^0 (w)$ and the local representation for linear combinations of these functions (and for the functions themselves if matrix $k_{ij}$ is nondegenerate) in terms of "potentials" $q^A$:
\beq
d_{v}(k_{ij}{\tilde F}^{A}_{i}dw^j)=0 \Rightarrow^{locally}\ k_{ij}{\tilde F}^{A}_{i}=\frac{\partial q^A}{\partial w^j}.
\eeq
Next we will see that the situation existing for the quadratic density $h^0(w)$ extends to the general concave (and convex) by $w$ function $h^{0}(w).$
\end{enumerate}
\end{example}

Let now $h^0$ be a \emph{concave function}, i.e. let the Hessian of $h^0$ be negative definite in its
domain
\[
\frac{\partial^2 h^0}{\partial w^{i}\partial w^{j}}(w)<0.
\]
Then, the transformation $w^i\rightarrow \lambda^j = \frac{\partial h^0}{\partial
w^{i}}\vert_{w^{k}=F^{0}_{k}(y)}$ is the \emph{globally defined} diffeomorphism. Using $\lambda^i$ instead of
$w^i=F^{0}_{i}(y)$ as the new variables along the fiber of the bundle $\pi:Y\rightarrow X$ we write
integrability condition (6.2) in the form

\beq d\lambda^i \wedge d_{v}{\hat F}^{A}_{i}(\lambda )=0, \eeq
where ${\hat F}^{A}_{i}(\lambda )={\tilde F}^{A}_{i}(w(\lambda )).$ Using Lemma 2 we see that this
condition is (locally ) equivalent to the existence of the functions (potentials) $h^A (\lambda)$ defined
uniquely, up to an addition of functions of $x$ such that

\beq {\hat F}^{A}_{i}(\lambda )=\frac{\partial h^A}{\partial \lambda^{i}}. \eeq

\begin{definition} Let $h^0 (w)$ be a concave function of variables $w^j$ such that the condition (6.4) is
satisfied.  Then the n-form
\[
\mathbf{h}=h^{\mu}\eta_{\mu}=h^0 \eta_{0}+h^A \eta_{A}
\]
for which condition (6.4) is fulfilled is called a (local) \textbf{\emph{n-potential of the balance system}} $\bigstar$.
Balance law (3.1) corresponding to the function $h^0$ will be said to be generated by the \textbf{n-potential
$\mathbf{h}$},\cite{RS}.
\end{definition}

In a case where $h^0$ is a nondegenerate at a point $(x,w_{0}=F^{0}(x, y_{0}))$ in the sense that its
vertical Hessian
\[
Hess(h^{0})(x,w_{0})=\left( \frac{\partial^2 h^0}{\partial w^{i}\partial w^{j}}\right)\vert_{x,w=w_{0}}
\]
is the nondegenerate matrix, the change of variables $\{w^{k}\}\rightarrow \{ \lambda^{j}=\frac{\partial
h^0}{\partial w^{j}}\}$ exists locally, in a neighborhood of the point $(x,w_{0})$ and the conclusions
done for the convex $h^0$ will follow locally.  As a result we get the following

\begin{theorem} Let in a neighborhood of a point $y_{0}$ the differentials $d_{v}F^{0}_{i}(y)$ form a
coframe in the fibers $Y_{x}$, i.e. functions $w^{k}(y)=F^{0}_{k}(x,y)$ are functionally independent.
 Then,
\begin{enumerate}
\item Locally, in a neighborhood of the point $y_{0}$ there is a bijection between the supplementary balance
  laws (3.1) for the balance system $\bigstar$ such that the main fields
  $\lambda^{i}=\frac{\partial K^0}{\partial w^i}$ are
functionally independent \emph{and the functions} $h^{0}(w^{k})$ of variables
  $x,w^{k}=F^{0}_{k}(y)$ with the nondegenerate vertical Hessian at the point
  $w_{0}=F^{0}_{k}(y_{0})$ defined up to an addition
  of an arbitrary function of $x$ and \textbf{such that} for the variables
  $\lambda^{i}= \frac{\partial h^0}{\partial
w^{i}}\vert_{w^{k}=F^{0}_{k}(y)}$ that are defined in a neighborhood of the point $y_{0}$ the
integrability conditions

\beq d\lambda^i \wedge d_{v}{\hat F}^{A}_{i}(\lambda )=0\Leftrightarrow d_{v}({\hat F}^{A}_{i}d\lambda^i
)=0
 ,\ A=1,2,\ldots,n,
 \eeq
where ${\hat F}^{A}_{i}(\lambda )=F^{A}_{i}(x,y(w(\lambda )))$, holds.

  \item For a given supplementary balance law (3.1) the function $h^0(x,w)$ is defined by the condition
  \[
h^{0}(x,w)=K^{0}(x,y(w)),\ \text{where} \ w^{k}(y)=F^{0}_{i}(x,y),
  \]
conditions (6.4) are fulfilled, local vertical coordinates $\lambda^{i}= \frac{\partial h^0}{\partial
w^{i}}\vert_{w^{k}=F^{0}_{k}(y)}$ that are defined in a neighborhood of the point $y_{0}$ satisfy to the
integrability condition and there exists functions $h^A (x,\lambda^i)$ such that
\[
{\hat F}^{A}_{i}(\lambda )={\tilde F}^{A}_{i}(w(\lambda ))=\frac{\partial h^A}{\partial \lambda^{i}},
\] thus defining the n-potential $h(\lambda)=h^{\mu}(x,\lambda)\eta_\mu ,$
corresponding to the balance law (3.1).
  \item Vice versa, for a given a given function $h^{0}(w^{k})$ of variables
 $x,w^{k}=F^{0}_{k}(y)$ satisfying to the conditions above at the point $w_{0}=F^{0}_{k}(y_{0})$,
the supplementary balance is defined from the conditions:
\beq
\begin{cases}
K^{0}(y)&=h^{0}(F^{0}_{k}(x,y)),\\
d_{v}K^{A}(y)&=\sum_{i}\lambda^{i}d_{v}F^{A}_{i},\ \forall A=1,\ldots,n,\\
Q(y)&=\sum_{i}\lambda^{i}\Pi_{i}(y)+K^{\mu}_{,x^\mu},
\end{cases}
\eeq where Lagrange multipliers $\lambda^{i}$ defined by $\lambda^{i}(y)=\frac{\partial
h^0}{\partial w^{i}}\vert_{w^{k}=F^{0}_{k}(y)}$ are functionally independent.
  \end{enumerate}
\end{theorem}
\begin{proof} Al statements except the last one were proved before the theorem.  Last one follows directly
from the Theorem 3.
\end{proof}

\begin{remark} To compare results obtained in this section with these of the RET in Sec.5 we have to
replace $h^\mu$  by $K^\mu$ here, ${\hat h}^{\mu}$ in Section 5 by $h^\mu$ here.
\end{remark}

\section{\textbf{Examples.}}
 \begin{example} Consider the case of one field $u(x,t)$ in the $1+1$ space-time.  Balance system takes the form
 \beq
 u_t +c(u)_{,x}=\Pi(t,x,u).
 \eeq
 System of LL-equations takes the form
 \beq
 \begin{cases}
 \lambda \cdot 1=K^0,\\
 \lambda c'(u)  = K^1.
 \end{cases}
 \eeq
 Substitute $\lambda$ from the first equation into the second and get
 \beq
 c'(u)K^{0}_{,u}=K^{1}_{,u}.
 \eeq
 Integrating we get
 \[
 K^{1}=\int^{u}c'(u)K^{0}_{,u}du+\tilde{K}^{1}(x,t)=c'(u)K^{0}- \int^{u}c''(u)K^{0}du+\tilde{K}^{1}(x,t).
 \]
 Thus, we have proved the following
 \begin{proposition} Let
 \beq
 u_t +c(u)_{,x}=\Pi(t,x,u)
 \eeq
 be a balance equation for scalar function $u(t,x)$ in (1+1)-dim space-time.  Then an arbitrary supplementary balance law for this equation has the form
 \[
 d_{t} K^{0}+d_{x}K^1 =Q,
 \]
 where
 \begin{enumerate}
 \item Density $K^0 (t,x,u)$ is an arbitrary function of arguments $t,x,u$,
 \item Main field (only one here) is $\lambda =K^{0}$,
 \item  Flux $K^1$ is given by expression
 \[
 K^{1}=c'(u)K^{0}- \int^{u}c''(u)K^{0}du+\tilde{K}^{1}(x,t),
 \]
 where ${\tilde K}(t,x)$ is arbitrary,
 \item Source/production term $Q$ is
 \beq
 Q = K^{0}\cdot \Pi +K^{0}_{,t}+K^{1}_{,x}.
 \eeq
 \end{enumerate}
 \end{proposition}
 \end{example}

\begin{example} As the second example we consider the Cattaneo heat propagation. Original basic fields are temperature $y^0=\vartheta$ and the heat flux $y^A=q^A, A=1,2,3$.  Balance equations
\beq
\begin{cases}
\partial_{t}(\rho \epsilon )+\partial_{x^A}(q^A )=0,\\
\partial_{t}(\tau q^B)+\partial_{x^A}(\delta^{A}_{B}\Lambda (\vartheta))=-q^B.
\end{cases}
\eeq
Here $\tau =\tau(\vartheta)\leqq 0$ is the relaxation time and $\Lambda (\vartheta)$ is the function such that in the expression $\partial_{x^B}(\Lambda (\vartheta)=\lambda(\vartheta)\frac{\partial \vartheta}{\partial x^B}$, $\lambda (\vartheta)$ - heat ... coefficient.  We assume that $\rho-const$.
\par
For simplicity we consider here the case where $\tau-const$, more general case is considered below, in Sec.12. If $\tau-const$, second equation can be written in the form
\[
\partial_{t}( q^B )+\partial_{x^A}(\delta^{B}_{A}\tau^{-1}\Lambda (\vartheta))=-\tau^{-1}q^B.
\]
Introduce variables $w^0 =\tilde \epsilon =\rho \epsilon,\ w^A=q^A$. In these variables, the balance system (7.6) takes the form
\beq
\begin{cases}
\partial_{t}(\tilde \epsilon )+\partial_{x^A}(q^A )=0,\\
\partial_{t}( q^B )+\partial_{x^A}(\delta^{B}_{A}\tilde{\Lambda} (\tilde \epsilon, q^C))=-\tau^{-1}q^B
\end{cases}
\eeq
with the new function ${\tilde \Lambda} ({\tilde \epsilon}, q^C)=\tau^{-1}\Lambda (\vartheta)$.\par
We have: ${\tilde F}^{A}_{0}=q^A,\ {\tilde F}^{A}_{C}=\delta^{A}_{C}{\tilde \Lambda}$.\par
System of equations (6.4) takes in these notations, the form
\begin{multline}
[(q^A_{,w^k}\partial_{w^j}\partial_{\te}-q^A_{,w^j}\partial_{w^k}\partial_{\te} +(\delta^{C}_{A}{\tilde \Lambda})_{,w^k}\partial_{w^j}\partial_{q^C}-(\delta^{C}_{A}{\tilde \Lambda})_{,w^j}\partial_{w^k}\partial_{q^C}]h^0=0,\\ A,C=1,2,3;k\ne j;k,j=0,1,2,3.
\end{multline}
Take, at first, $k=0,j=B>0$ and, correspondingly, $w^k=\te, w^B=q^B$. Then, previous system of equations takes the form of subsystem I.
\beq
I:\ [-\delta^{A}_{B}\partial^{2}_{\te}+{\tilde \Lambda}_{,\te}\partial_{q^B}\partial_{q^A}-{\tilde \Lambda}_{,q^B}\partial_{\te}\partial_{q^A}]h^0 =0,\ \forall A,B=1,2,3.
\eeq
Consider separately two cases
\beq
\begin{cases}
A=B,\ [-\partial^{2}_{\te}+{\tilde \Lambda}_{,\te}\partial^{2}_{q^B}-{\tilde \Lambda}_{,q^B}\partial_{\te}\partial_{q^B}]h^0 =0,\\
A\ne B,\ [{\tilde \Lambda}_{,\te}\partial_{q^B}\partial_{q^A}-{\tilde \Lambda}_{,q^B}\partial_{\te}\partial_{q^A}]h^0 =0.
\end{cases}
\eeq
Take now the rest of cases (we remind that $k\ne j$ and that the system (7.8) is antisymmetric by $(kj)$): $k=D\ne B=j,\ w^k=q^D,w^j=q^B$. For this case, equations (6.14) takes the form of second subsystem of (7.8):

\begin{multline}
II:\ [\delta^{A}_D \partial_{q^B}\partial_{\te}-\delta^{A}_B\partial_{q^D}\partial_{\te} +\delta^{C}_{A}{\tilde \Lambda}_{,q^D}\partial_{q^B}\partial_{q^C}-\delta^{C}_{A}{\tilde \Lambda}_{,q^B}\partial_{q^D}\partial_{q^C}]h^0=0,\\ A,B,C,D=1,2,3;D\ne B.
\end{multline}
Consider two cases:
\beq
\begin{cases}
A\ne C:\  [\delta^{A}_D \partial_{q^B}\partial_{\te}-\delta^{A}_B\partial_{q^D}\partial_{\te}]h^0 =0,\\
A=C:\ [\delta^{A}_D \partial_{q^B}\partial_{\te}-\delta^{A}_B\partial_{q^D}\partial_{\te} +{\tilde \Lambda}_{,q^D}\partial_{q^B}\partial_{q^A}-{\tilde \Lambda}_{,q^B}\partial_{q^D}\partial_{q^A}]h^0=0.
\end{cases}
\eeq
Let indices $A=D,B,C$ are all different. Then the first equation in (7.12) takes the form
\beq
\partial_{q^B}\partial_{\te}h^0=0.
\eeq
For arbitrary $B=1,2,3$ we can choose A,C such that the condition that indices $A=D,B,C$ are all different holds.  Thus, the condition (7.12) is fulfilled \textbf{FOR ALL B}.  Then, the first equation in (7.12) is fulfilled due to this condition.\par
  At the same time, second equation in (7.12) takes the form
\beq
[{\tilde \Lambda}_{,q^D}\partial_{q^B}\partial_{q^A}-{\tilde \Lambda}_{,q^B}\partial_{q^D}\partial_{q^A}]h^0=0,\ B\ne D.
\eeq
Equations (7.10) now take the form
\beq
\begin{cases}
\ [-\partial^{2}_{\te}+{\tilde \Lambda}_{,\te}\partial^{2}_{q^B}]q^0 =0,B=1,2,3,\\
A\ne B,\ {\tilde \Lambda}_{,\te}\partial_{q^B}\partial_{q^A}q^0 =0.
\end{cases}
\eeq
From the condition $\Lambda_{,\vartheta}\ne 0$ it follows that ${\tilde \Lambda}_{,\te}\ne 0$. As a result, second of the equations (7.15) is equivalent to the condition
\beq
\partial_{q^B}\partial_{q^A}q^0 =0,\ A\ne B.
\eeq
In equation(7.14) indices $A,B,D$ can not coincide.  On the other hand, if all three of them are different, then (7.15) is fulfilled due to (7.16).  This left the case $A=B\ne D$ (and symmetrical case).  In this case $A\ne D$ and (7.15) reduces to
\beq
{\tilde \Lambda}_{,q^D}\partial^{2}_{q^B}h^0=0,\ B\ne D.
\eeq
As a result, equations (7.8) reduces to the following system of conditions:
\beq
\begin{cases}
{\tilde \Lambda}_{,q^D}\partial^{2}_{q^B}h^0=0,\ B\ne D,\\
\partial_{q^B}\partial_{q^A}q^0 =0,\ A\ne B,\\
\partial_{q^B}\partial_{\te}h^0=0,\\
[-\partial^{2}_{\te}+{\tilde \Lambda}_{,\te}\partial^{2}_{q^B}]q^0 =0,B=1,2,3.
\end{cases}
\eeq
From the second and third condition and from the invariance of $h^0$ under the spacial rotations it follows that
\beq
h^0=\gamma(\te)+\alpha \Vert \bar{q}\Vert^2 +\beta_{A}q^A,
\eeq
where $\alpha ,\beta_{A}$ do not depend on $\te$ and $\beta_A$ is the spacial covector.\par
In the first condition requirement ${\tilde \Lambda}_{,q^D}=0$ for all $D$ is equivalent to the requirement that $\te_{,q^C}=0$ for all $C$.
Thus, if $\te$ does not depend on the heat flux, then ${\tilde \Lambda}_{,q^D}=0$ and the first condition in (7.18) is fulfilled.  Then the last equation in (7.18) gives, for $h^0$ in (7.19)
\[
\gamma''(\te)=2\alpha {\tilde \Lambda}_{,\te}\leftrightarrow \gamma(\te)=2\alpha \int^{\te}{\tilde \Lambda}d\te +c\te +d.
\]
Thus, in the case where $\te_{,q^C}=0$,
\beq
h^0=\left[2\alpha \int^{\te}{\tilde \Lambda}d\te +\alpha \Vert \bar{q}\Vert^2\right] +c\te +\beta_{A}q^A+d.
\eeq
Now consider the case where ${\tilde \Lambda}_{,q^D}\ne 0$, then $\partial^{2}_{q^B}h^0=0$ for all $B$. Using this in (7.20) we see that in this case $\alpha =0$.  Thus, we can conclude that
\begin{proposition} For the Cattaneo heat propagation model with $\tau_{,\vartheta}=0$,
\begin{enumerate}
\item If $\te_{,q^C}=0,\ \forall C$, then
\[
h^0=\left[2\alpha \int^{\te}{\tilde \Lambda}d\te +\alpha \Vert \bar{q}\Vert^2\right] +c\te +\beta_{A}q^A+d,
\]
\item If $\te_{,q^C}\ne 0$, then
\[
h^0=c\te +\beta_{A}q^A+d.
\]
\end{enumerate}
\end{proposition}
In the expression (7.20), last term corresponds to the secondary balance laws that do not depend on the fields (trivial), previous two terms correspond to the initial balance system (7.6) but expression in brackets gives the density of new balance law.  See below, Sec.12 for more details.
\end{example}
\begin{remark}
Form of $h^0$ obtained here does not allow to determine the fluxes $h^0$ - for this one has to solve system of equations for $h^A$.  In Sec.12 the whole LL-system of equations will be solved and, together with the description of all supplementary balance laws, corresponding constitutive limitations on the form of internal energy $\epsilon =\epsilon(\theta ,q)$ will be obtained.
\end{remark}

\section{\textbf{Example: (2+2)-RET balance systems.}}
\begin{example} Consider a model example of a balance system with one space variable $x$ and two fields
$y^1 ,y^2$.  The balance system has the form

\beq
\partial_{t}y^i +\partial_{x}F^{1}_{i}=\Pi_{i},\ i=1,2.
\eeq

In such a case the system of equations for $h^0$ reduces to one equation with $k=1,j=2$:
\[
\left( F^{1}_{i,y^1}\partial_{y^i}\partial_{y^2}-F^{1}_{i,y^2}\partial_{y^i}\partial_{y^1}\right)h^0 =0,
\]
or \beq \left[ -F^{1}_{i,y^2}\partial^{2}_{y^1
 y^1}+(F^{1}_{1,y^1}-F^{1}_{2,y^2})\partial_{y^1}\partial_{y^2}+
F^{1}_{2,y^1}\partial^{2}_{ y^2  y^2} \right]h^0=0. \eeq
\end{example}
This is, for each point $(t,x)$ the homogeneous linear equation of second order with two independent
variables. Space of solutions of this equations is \emph{infinite-dimensional}.\par
\vskip1cm
Next we study the question  - which of equations 8.2 have solutions $h^0$ with definite (positive or negative) Hessian $Hess(h^0 )=\left( \frac{\partial^2 h^0}{\partial y^i \partial y^j}\right)$.\par
Near a point where type of this equation is locally constant, this equation can be transformed,
by a change of variables, to one of three canonical types - Laplace equation (elliptic), degenerate
(parabolic) and wave (hyperbolic).  The type at a point $y_{0}=(x_{0},y_{0})$ is determined by the
determinant of the matrix

\beq
J(x,u)=\begin{pmatrix} -F^{1}_{1,y^2} & \frac{1}{2}(F^{1}_{1,y^1}-F^{1}_{2,y^2})\\
\frac{1}{2}(F^{1}_{1,y^1}-F^{1}_{2,y^2}) & F^{1}_{2,y^1}
\end{pmatrix}.
\eeq
\begin{enumerate}
\item Elliptic case: $Det(J)>0.$  In such a case, equation (8.2) is locally isomorphic to the standard Laplace
equation on the plane.  Its solutions are harmonic functions of $y^i$. Hessian matrix of a harmonic
function has zero trace - by the Laplace equation itself.  As a result, sum of its eigenvalues - trace of
the matrix is zero and a function $h^0$ constructed in such a way can not be positive or negative definite
in canonical coordinates. Yet, the change of variables that transformed equation (8.2) to the Laplace
equation is not tensorial and additional terms may change the sign of trace.  Yet, let us fix a point $y$
and make a linear change of variables that transform the quadratic form of equation (8.2) \textbf{at the
point $y$} to the canonical (Laplace) form.  In new coordinates $(w^1 ,w^2 )$ \textbf{at the point $y$} we
will have $h^{0}_{w^1 w^1}+h^{0}_{w^1 w^1}=Tr(Hess(h^{0})(w^i )=0.$  Now we use the fact that under a
linear transformations of variables the Hessian of a function transforms tensorially and, as a result, its
trace is invariant under the linear change of variables.  As a result, in original variables
$Tr(Hess(h^{0}))(y^i )=0.$  Thus, eigenvalues of this matrix have (nonzero) opposite signs and the Hessian
can not be positive or negative definite.
\item  Hyperbolic case: $Det(J)<0$.  Equation (8.2) by a local change of variables can be reduced to the
standard wave equation $(\partial^{2}_{v^2}-c^2\partial^{2}_{w^2})f=0.$ General solution of this equation
\[
f=f_{1}(w-cv)+f_{2}(w+cv)
\]
has the Hessian
\[
Hess(f)=\begin{pmatrix} c^2(f_{1}''+f_{2}'') & c(-f_{1}''+f_{2}'')\\
c(-f_{1}''+f_{2}'') & f_{1}''+f_{2}''
\end{pmatrix}.
\]
Determinant of Hessian is equal to $Det(Hess(h^{0})=4c^2 f_{1}'' f_{2}'' $. Thus, if functions
$f_{1},f_{2}$ have the same convexity type, both eigenvalues of Hessian (if they are nonzero) have the
same sign and the Hessian form is \textbf{definite}, if functions $f_{1},f_{2}$ have the opposite
convexity type , eigenvalues of Hessian (if they are nonzero) have the opposite sign and the Hessian form
is indefinite.

\item Degenerate case $Det(J)=0$.  In such a case the equation can be reduced to a parabolic
equation $(\partial_{w}-\partial^{2}_{v^2})h^0=0,$ or degenerate $\partial^{2}_{v^2}h^0=0$. For a
degenerate case any solution has the form $h=h_{1}(w)+h_{2}(w)v$ and its Hessian
\[
Hess(h)=\begin{pmatrix} h_{1}''+vh_{2}'' & h_{2}'\\
h_{2}' & 0
\end{pmatrix}.
\]
Thus, characteristic equation has the form
\[
\lambda^2 -(h_{1}''(w)+vh_{2}''(w))\lambda -h_{2}'(w)^2=0.
\]
Its roots
\[
\lambda_{\pm}=-\frac{1}{2}(h_{1}''(w)+vh_{2}''(w))\pm\sqrt{\frac{1}{4}(h_{1}''(w)+vh_{2}''(w))^2
+h_{2}'(w)^2}
\]
are also of the opposite sign and, therefore, Hessian matrix of $h^0$ can not be positive or negative
definite. Using the linear change of variables as for elliptic case and the fact that under such a change
of variables the homogeneous second order equation (8.2) transforms into the homogeneous second order
equation we get, in new coordinates $w^i$ equation of the same type as (8.2) such that at the chosen point
$y$ it has the form $h^{0}_{w^2 w^2}=0$. It follows from this that at this point $Det(Hess(h^0))(w)\leqq 0
$.  By the same tensorial arguments as for elliptic case this will be true at the point $y$ in
$y^i$-variables as well ($Det(Hess(h^0 ))$ is multiplied by the square of the Jacobian of linear
transformation and its sign is not changing) and the Hessian form $Hess(h^0 )(y)$ can not be positive or
negative definite (but can be positive or negative semi-definite !).
\end{enumerate}
These arguments prove the following
\begin{proposition} Let a balance system with one space variable $x$ and two fields
$y^1 ,y^2$ has the form

\[
\partial_{t}y^i +\partial_{x}F^{1}_{i}=\Pi_{i},\ i=1,2.
\]
If there is a (secondary) balance equation (2.1) with positive or negative definite $Hess(K^{0}(y))$, then
the equation (8.2) is hyperbolic:
\[
Det(J(u))=-F^{1}_{1,y^2 }F^{1}_{2,y^1}-\frac{1}{4}(F^{1}_{1,y^1}-F^{1}_{2,y^2 })^2<0.
\]
\end{proposition}
Let now the balance system (17.1) is  such that the defining equation (8.2) is hyperbolic, i.e $Det(J(y))<0$. Introduce
notations $a=-F^{1}_{1,y^2},\ 2b=(F^{1}_{1,y^1}-F^{1}_{2,y^2}),\ c=F^{1}_{2,y^1}$.  Then the equation
(17.2) for $h^0$ takes the form
\[
(a\partial^{2}_{y^2\ 2}+2b\partial^{2}_{y^1 y^2}+c\partial^{2}_{y^1\ 2})h^0 =0.
\]
Introduce the characteristic functions $\lambda_{1,2}$ of this equation as solutions of the quadratic
equation
\[
a\lambda^2 +2b\lambda +c=0.
\]
Let functions $\phi(y^1,y^2)$ and $\psi(y^1 ,y^2 )$ are Riemann invariants (see \cite{CG}), i,e, solutions
of the equations
\[
\begin{cases} \phi_{,y^1}-\lambda_{1}\phi_{,y^2}&=0,\\
\psi_{,y^1}-\lambda_{2}\psi_{,y^2}&=0.
\end{cases}
\]
Change of variables
\[
\begin{cases}
\xi &=\phi(y^1 ,y^2 ),\\
\eta &=\psi (y^1 ,y^2)
\end{cases}
\]
reduces the equation for $h^0$ to the canonical form
\[
h^{0}_{,\xi \eta}=0,
\]
having general solution of the form
\[
h^{0}=h_{1}(\xi )+h_{2}(\eta ).
\]
In variables $y^i$ (locally) general solution of last equation has the form
\[
h^0 (y^i)=h_{1}(\phi(y^1 ,y^2 ))+h_{2}(\psi(y^1 ,y^2 )).
\]
Calculating Hessian of this function we get
\begin{multline}
Hess (h^0 )=h_{1}^{'}(\phi(y^1 ,y^2 ))\begin{pmatrix} \phi_{y^1 y^1 }& \phi_{y^1 y^2}\\
\phi_{y^1 y^2}& \phi_{y^2 y^2 }\end{pmatrix}+h_{2}^{'}(\psi(y^1 ,y^2 ))\begin{pmatrix} \psi_{y^1 y^1 }& \psi_{y^1 y^2}\\
\psi_{y^1 y^2}& \psi_{y^2 y^2 }\end{pmatrix}+\\ +h_{1}^{''}(\phi(y^1 ,y^2 ))
\begin{pmatrix} \phi^{\ 2}_{y^1 }& \phi_{y^1 }\phi_{y^2}\\
\phi_{y^1 }\phi_{y^2}& \phi^{\ 2}_{y^2}\end{pmatrix}+ h_{2}^{''}(\psi(y^1 ,y^2 ))
\begin{pmatrix} \psi^{\ 2}_{y^1 }& \psi_{y^1 }\psi_{y^2}\\
\psi_{y^1 }\psi_{y^2}& \psi^{\ 2}_{y^2}\end{pmatrix}.
\end{multline}
In this expression third and forth matrices are positive semi-definite -one eigenvector of each of them is
zero, the other ($\lambda^{2}_{i}+1$) is positive. If the graphs of Riemann invariants $\phi ,\psi $ are
convex, then one can construct examples of $h^0$ with positive or negative definite Hessian by the choice
of $h_{1},h_{2}$ having first and second derivatives of the same sign. For instance, if the Hessian of
$\phi$ is positive definite, one can take $h_{2}=0$ and $h_{1}$ to be monotonically decreasing and convex
function.

Algorithmically this approach allows to choose functions $h_{i}$ for which the solution $h^0 $ has the
negative definite Hessian.

\begin{example} Let $\phi $ be any harmonic function of its variables.  Let $\chi $ be the harmonically
conjugate to $\phi $ (i.e. function $(\phi +i \chi )(y^1 +iy^2 )$ is analytic.  Then, due to the
Cauchy-Riemann equations
\[
Hess(\phi )=\begin{pmatrix} \chi_{y^1 y^2} & \phi_{y^1 y^2}\\
\phi_{y^1 y^2} & -\chi_{y^1 y^2}
\end{pmatrix},
\]
and $Det (Hess (\phi ))=-\chi^{2}_{y^1 y^2}-\phi^{2}_{y^1 y^2}\leqq 0.$  Choosing $h_{1}$ appropriately,
one can guarantee the existence of a function $h^0$ with the nonpositive Hessian.
\end{example}

\section{\textbf{Integrability condition and the D-module $\mathcal{M_C}$}.}
Integrability condition (6.4) represents an \textbf{overdetermined system of linear equations of second order} for the
function $h^0$. Therefore, except of the obvious linear by $w^i$ solutions might exist only under the certain
conditions on the current part of the constitutive relations.\par
  Consider a special case where densities $F^{0}_{i}(x,y)$ coincide with the basic fields $y^i$.
  In this case $w^{i}=y^{i}$ and condition (6.4)
\[\frac{\partial^2 h^0}{\partial y^i \partial y^j}dy^j\wedge d_{v}F^{A}_{i}(y)=0,\ A=1,\ldots ,n \]

takes the form
 \beq
F^{A}_{i,y^k}h^{0}_{,y^i y^j }dy^j \wedge dy^k =0,\ A=1,2,\ldots ,n,
 \eeq
or \beq F^{A}_{i,y^{[k}}\frac{\partial^2}{\partial y^{\vert i\vert} \partial y^{j]}}h^{0}=0,\ \text{for\
all} \ k\ne j,A=1,2,\ldots ,n. \eeq

 This system has all the functions $h^{0}=\sum_{k}c_{k}(x)y^k$ linear by $y^i$ as
trivial solutions, corresponding to the linear combination of balance equations of the system
($\bigstar$). Generically there are no other solutions.  This has its reflection in the known fact that
the entropy principle - existence of a nontrivial solution of this system with the negative definite
Hessian place strong restrictions to the form of the constitutive relations $\mathcal{C}$.\par

Introduce the vertical vector fields $\eta^{A}_{k}=F^{A}_{i,y^k}(x,y)\partial_{y^i}$, then the system
(9.2) takes the form

\beq (\eta^{A}_{k}\partial_{y^j}-\eta^{A}_{j}\partial_{y^k})h^0 =0,\ k\ne j,A=1,2,\ldots ,n. \eeq

This system determines the cyclic $\mathcal{D_{E}}$-module $\mathcal{M}_{C}$ over the ring $\mathcal{E}$
of smooth functions of $y^i$ containing the coefficients $F^{A}_{i}$ (\cite{B}).  Let now the components
$F^{\mu}_{i}$ of the current are analytic functions. Consider their extension to the complexification
$Y^{c}$ of the manifold $Y$  (see \cite{B,Ka}).  This allows us to extend $\mathcal{D_{E}}$ to the
$\mathcal{D}_{\mathcal{O}}$-module where $\mathcal{O}(Y^c )$ is the sheaf of holomorphic functions on
$Y^c$. Choosing the admissible filtration of the algebra $\mathcal{D_{E}}$ we get the graded module
$gr(\mathcal{M}_{\mathcal{C}})=\mathcal{E[\zeta_j ]}/I$ where the ideal $I$ is generated by the quadratic
polynomials
\beq
P^{A}_{ij}(\zeta) =\eta^{A}_{i}(\zeta)\zeta_{j}-\eta^{A}_{j}(\zeta)\zeta_{i}=F^{A}_{k,y^i}(x,y)\zeta_k \zeta_j -F^{A}_{k,y^j}(x,y)\zeta_k \zeta_i.
\eeq
The characteristic variety $Char(\mathcal{M_{C}})$ of $D$-module $\mathcal{M_{C}}$ (more exactly, the
relative characteristic variety with respect to the projection $Y\rightarrow X$ (see \cite{B}, Sec.1.6) is
defined as the support of the $\mathcal{E[\zeta_j ]}$-module $gr(\mathcal{M_{C}})$ - i.e. as the
coisotropic subset of the (vertical) cotangent bundle $V^{*}(Y)\rightarrow Y$ where
\[
P^{A}_{ij}(\zeta)=0, \text{for\ all}\  i,j,A.
\]

 The set of expressions  in the right side of (9.4) can be considered as the collection of coefficients of the
2-form
\[
(\eta^{A}_{i}(\zeta)dq^i)\wedge (\zeta_{j}dq^j )
\]
with the coefficients being functions in $Y$. This wedge product is zero if and only if the 1-forms are
proportional, i.e. if
\[
\eta^{A}_{i}(\zeta)=\mu^A (x,y)\zeta_{i}, A=1,\ldots,n;i=1,\ldots,m
\]
with some factors $\mu^{A}(x,y)$.  Left side of this equation can be considered as the result of
application of linear operator $L^A$ with the matrix $L^{A\ i}_{j}= F^{A}_{j,y^i}(x,y)$ to the (co)vector
 with components $\zeta_{j}$.  Then the last equality takes the form
\[
L^{A}{\vec \zeta}=\mu^{A}(y){\vec \zeta}.
\]
Thus, a (vertical) covector $(y,\zeta )$ belongs to the support of the module $gr(\mathcal{M_{C}})$ if and
only if $\zeta$ is the common (real) eigenvector of operators $L^{A}(y)$ for all $A=1,\ldots ,n.$
\emph{Generically} there are no nonzero vector with this property and $Char(\mathcal{M_{C}})(x)=U_{x},\ x\in X$ coincide with the zero section of the (vertical) cotangent bundle.
\par
For each $A=1,\ldots, n$ let
\[P^{A}(x,y)(\eta)=\prod_{\mu^j (x,y)\in \sigma(L^A )} (\eta -\mu^j (x,y))\]
to be the reduced characteristic polynomial of endomorphism $L^A$. Product here is taken over all
different eigenvalues $\mu^j$ of the endomorphism $L^{A}(x,y)$. It is easy to see that the operator
$P^{A}(x,y)(L^A )$ annulate all the eigenvectors of $L^{A}(x,y)$ and no other vectors. To see it let
$E(\mu)$ be a root subspace corresponding to the eigenvalue $\mu \in \sigma(L^A )$. Then any operator $L^A
-\mu^j (x,y)I$ in the product above with $\mu^j\in \sigma(L^A ), \mu^j \ne \mu$ leaves the subspace
$E(\mu)$ invariant and is invertible on this subspace. Since factors in this product commute between
themselves the kernel of $P^{A}(x,y)(L^A )$ restricted on $E(\mu)$ is the subspace of eigenvectors with
eigenvalue $\mu$. Therefore, condition for a vector $\zeta \in T^{*}_{x,y}(U^{c}_{x,y})$ to belong to the
characteristic variety of the $\mathcal{D_{O}}$-module $\mathcal{M_{C}}$ is equivalent to the fulfillment
of the system of linear relations

\beq P^{A}(x,y)(L^A )\zeta=0,\ A=1,\ldots, n. \eeq
As a result we have proved the following
\begin{theorem}
\begin{enumerate}
\item Let $\mathcal{C}$ be a a constitutive relation of the RET type with $F^{0}_{i}=y^i$. Then a function $h^{0}(y^i)$
generates a supplementary balance law with the functionally independent Lagrange multipliers (main fields) $\lambda^i$
in a neighborhood of a point $(x_{0},y_{0})$ if and only if it is a solution of the cyclic D-module
$\mathcal{M_{C}}$
\[
(\eta^{A}_{k}\partial_{y^j}-\eta^{A}_{j}\partial_{y^k})h^0 =0,\ k\ne j,A=1,2,3,
\]
where $\eta^{A}_{k}=F^{A}_{i,y^k}(x,y)\partial_{y^i}$ and such that the (vertical) Hessian
\[ Hess(h^{0})(x_{0},y_{0})=\left( \frac{\partial^2 h^0}{\partial y^{i}\partial
y^{j}}\right)\vert_{x=x_{0},y=y_{0}}\] is nondegenerate in a neighborhood of a point $(x_{0},y_{0})$.
\item
Characteristic variety of the cyclic $\mathcal{D_{E}}$-module $\mathcal{M_{C}}$ is the union of the zero
section of the vertical cotangent bundle $V^{*}(Y)\rightarrow Y$ and the set of lines generated by the
common eigenvectors of the linear operators $L^{A}(y):V^{*}_{y}(Y)\rightarrow V^{*}_{y}(Y),\  A=1,2,3$ defined by
\[
(L^{A}(y)\zeta )_j =F^{A}_{k,\ y^j }\zeta_k.\]
\end{enumerate}
\end{theorem}
Deeper study of the D-module  $\mathcal{M_{C}}$ and the corresponding secondary balance laws for a RET
system with a constitutive relation $C$ will be presented elsewhere.\par

\section{\textbf{Defining system for $h^0$.}}
In applications it rarely happens that $F^{0}_{i}=y^i$.  In such a case one has to use solvability conditions (6.4) written in terms of variables $w^i$:
\beq
\frac{\partial^2 h^0}{\partial w^i \partial w^j}dw^j \wedge {\tilde F}^{A}_{i,w^k}dw^k=0 \Leftrightarrow ({\tilde F}^{A}_{i,w^j}\partial_{w^k}\partial_{w^i}-{\tilde F}^{A}_{i,w^k}\partial_{w^j}\partial_{w^i})h^0 =0,\ j\ne k;\ A=1,2,3.
\eeq
This requires the reversal of relations $w^i=F^{0}_{i}(y^j )$ to get $y^j= G^{j}(w^k )$ and use them to determine the flux components $\tilde F^{A}_{i}(w)$ as functions of $w^i$.\par
It is easy to see that the statements of Theorem 5 are valid for this case. In particular, characteristic variety of the overdetermined system 10.1 is formed by the union of zero section of the dual vertical bundle $V^{*}(\pi)\rightarrow Y$ and the common eigenvectors of three linear operators $L^{A}(w):V^{*}_{w}(Y)\rightarrow V^{*}_{w}(Y),\  A=1,2,3$ defined by
\beq
(L^{A}(w)\zeta )_j =\tilde{F}^{A}_{k,\ w^j }\zeta_k.\eeq
If three matrices $L^{A}_{kj}=\tilde{F}^{A}_{k,\ w^j }$ \emph{are in generic position}, they do not have common eigenvectors and characteristic variety of system (10.1) coincide with the zero section of the vertical cotangent bundle. This proves the first and, therefore, second statements of the following
\begin{proposition}\begin{enumerate} \item Generically, characteristic variety of the defining system (10.1) coincide with the zero section of the vertical cotangent bundle $V^{*}\rightarrow Y$.
\item Generically, the overdetermined system (10.1) is \textbf{holonomic}.
\item Generically, in the case where constitutive relations $C$ are analytical, the space of solutions of (10.1) is finite-dimensional.
\end{enumerate}
\end{proposition}
Third statement follows from the second one and the finiteness of the space of solutions of a holonomic system (see \cite{Ka} or \cite{B}).

Rewrite condition (6.2) in the form
\[
d_{v}\left( \frac{\partial h^0}{\partial w^i}\right)\wedge d_{v}(F^{A}_{i})=0,\ A=1,2,3,
\]
or, in the form
\beq
d_{v}\left[ \left( \frac{\partial h^0}{\partial w^i}\right) d_{v}(F^{A}_{i}) \right]=0 \Leftrightarrow
d_{v}\left[ \left( \left( F^{A}_{i,j}\frac{\partial}{\partial w^i}\right) h^0\right) dw^j\right]=0,\ A=1,2,3.
\eeq

Last equation can be written in the form
\[
d_{v}\left[ (\eta^{A}_{j}\cdot h^0)dw^j\right]=0,
\]
using the vector fields $\eta^{A}_{j}=F^{A}_{i,j}\partial_{i}$ and the corresponded "twisted differential" $d^{FA}$ of the functions
\beq
d^{FA} \cdot f =(\eta^{A}_{j}\cdot f) dw^j.
\eeq
It can also be presented in the form
\beq
d_{v}(J^A d_{v}h^0)=0,
\eeq
where $J^A :V^{*}(\pi) \rightarrow V^{*}(\pi)$ is the endomorphism of the vertical cotangent bundle defined by $F^{A}_{i,j}$:
\[
J^A:\mu_{k}dw^k \rightarrow F^{A}_{i,j}\mu_i dw^j.
\]
\begin{remark} As it follows from the first form (10.3) of defining system, functions $h^0 =a_{i}w^i+a,a_{i},a \in \mathbb{R}$ are solutions of defining system for all balance systems.
\end{remark}
\begin{lemma} Twisted differential $d^{FA} :C^{\infty}(U_x)\rightarrow \Omega^{1}(U_{x})$ is elliptic iff $Ker(F^{A}_{\cdot ,\cdot})= 0$.
\end{lemma}
\begin{proof} It is easy to see that the symbol $\sigma(d^{A}F):U\otimes T^{*}(U)\rightarrow \Omega^{1}(U)$ has the form:
\[
\lambda =\{ \lambda^{i}\}\rightarrow F^{A}_{i,j}\lambda^i dw^j.
\]
Thus, a (constant) element $\lambda =\{ \lambda^{i}\}$ is characteristic for $d^{FA}$, at a point $w$ iff $F^{A}_{i,j}\lambda^i=0$ for all $j$.
\end{proof}

\begin{corollary} System (10.1) is elliptic if and only if intersection \[\bigcap_{A=1}^{3}Ker(F^{A}_{\cdot ,\cdot})=0.\]
\end{corollary}

System (10.1) has the $m+1$-dim space of "trivial" solutions - $h^{0}=a_{i}w^i +a$ where $a_{i},a$ are functions of $x^\mu$ only.

\begin{example} Systems of the form (10.1) may put a very diverse restrictions to the function $h^0$ and to the constitutive relations in terms of $F^{\mu}_{i}$.  To illustrate this diversity we consider the simplest case when $m=2$ so that there are only two variables $w^1,w^2$ and for all $A$, matrices $F^{A}_{i,j}$ \textbf{are constant}.\par
More then this, consider just one equation of the form (10.1) with the constant matrix $F_{ij}=F_{i,j}$  Thus, the equation for $h^{0}(w^1,w^2)$ has the form
\beq
d_{v}(F^{j}_{i} h^{0}_{,w^k}dw^k)=0.
\eeq
In the next Table we collect information about the from and type of solutions of equation (10.6) for different constant matrices $F_{ij}:$

\begin{table}[h]
\begin{center}
\begin{tabular}{|l|r|r|}
  \hline
  Matrix $F_{ij}$ &   Equation & Solutions ,\\ \hline
  $F_{ij}=\begin{pmatrix}a & 0\\0 & b \end{pmatrix},a\ne b$ & $h^{0}_{y^1 y^2}=0$                 & $h^{0}=f(y^1)+g(y^2)$, \\ \hline
  $F_{ij}=\begin{pmatrix}0 & a\\0 & 0 \end{pmatrix},a\ne 0$ & $h^{0}_{y^2y^2}=0$                  & $h^0=f(y^1) y^2+g(y^1)$, \\ \hline
  $F_{ij}=\begin{pmatrix}0 &1\\-1 & 0 \end{pmatrix}$        & $h^{0}_{,y^1y^1}+h^{0}_{y^2y^2}=0$  & Harmonic\  functions on the plane,   \\ \hline
  $F_{ij}=\begin{pmatrix}0 &1\\1  & 0 \end{pmatrix}$        & $h^{0}_{,y^1y^1}-h^{0}_{y^2y^2}=0$  & $h^0 =f(y^1 +y^2 )+g(y^1 -y^2)$.\\ \hline
$F_{ij}=\begin{pmatrix}a &b\\-b & -a \end{pmatrix}$        & $h^{0}_{,11}+2\frac{b}{a}h^{0}_{12}+h^{0}_{22}=0$  & Elliptic if $\frac{b^2}{a^2}<0$, hyperbolic if  $\frac{b^2}{a^2}>0$   \\ \hline
\end{tabular}
\caption{Type of solutions of equation 10.6. }\label{Ta}
\end{center}
\end{table}

From this table we see obvious relation between the type of equation and the action of 1-parameter subgroup of $SL(2,R)$ corresponding to matrices $M$.
\end{example}

\begin{example} Let $m$ be arbitrary and matrix $M=F^{A}_{i,j}$ - constant and diagonal: $M=diag(a_{i})$. Arrange diagonal terms so that $M$ would take block-diagonal form such  that  each block $M_l$ of size  $m_l\times m_l$ corresponds to the subset of variables $y^{j},\ j\in I_l$, has the form $a_{l}I_{m_l}$ and $a_{l}\ne a_{s},s\ ne l$.  Here $[1,m]=I_{1}\cup\ldots I_{k}$ is the splitting of the set of indices $[1,m]$ into non-intersecting subsets.   Then it is easy to see that each subsystem of the system (10.1), corresponding to the fixed value of $A$,  splits into the equations
\[a_i \delta^{i}_{j}h^{0}_{,y^j y^k}-a_i \delta^{i}_{k}h^{0}_{,y^j y^j}=0\Leftrightarrow (a_j-a_k)h^{0}_{,jk}=0.\]
General solution of this system has the form
\beq
h^0 =\sum_{l}h^{l}(y^j,j\in I_{l})
\eeq
where all the functions $h^{l}(y^j,j\in I_{l})$ are arbitrary functions of their variables.
\end{example}
\begin{example} Let $M=diag(f_{i}(y^1))$ - diagonal matrix with coefficients depending only on $y^1$. It is easy to see that in the case where diagonal elements $f_i$ of $M$ may depend on all the vertical variables, equations 10.1 have the form
\[ d(f_{i}h^{0}_{,i}dx^i)=0 \Leftrightarrow (f_{i,j}h^{0}_{,i}+f_{i}h^{0}_{,ij})dx^j\wedge dx^i =0.\]
This can be rewritten in the form
\beq
f_{i,j}h^{0}_{,i}-f_{j,i}h^{0}_{,j}=(f_{j}-f_{i})h^{0}_{,ij}, \ i\ne j.
\eeq
In the case, where coefficients $f_i$ depend on $y^1$ only, this system splits:
\beq
\begin{cases}
f_{i,1}h^{0}_{,i}=(f_{1}-f_{i})h^{0}_{,i1},\ i>1,j=1,\\
(f_{j}-f_{i})h_{,ij}=0,\ i,j>1.
\end{cases}
\eeq

First equation can be written in the form
\[
\frac{f_{i,1}}{f_{1}-f_{i}}=\frac{h^{0}_{,i1}}{h^{0}_{,i}}=(ln(\vert h^{0}_{,i} \vert))_{,1}.
\]
Integrate by $y^1$ and denote a antiderivative of $\frac{f_{i,1}}{f_{1}-f_{i}}$ by $F^{i}(y^1)$.  We get
\beq
ln(\vert h^{0}_{,i} \vert)=F^{i}(y^1)+G^{i}(\hat y^i
).\eeq
Function $G^{i}$ may depend on all the variables $y^i$ except $y^1$.\par
Last relation  can be rewritten in the form
\beq
h^{0}_{,i}=\pm e^{F^{i}(y^1)}e^{G^{i}(\hat y^i
)}.
\eeq
notice that $h^0=0$ is also solution of the first equation above.  As a result, we get
\beq
h^{0}_{,y^i}=\phi^{i}(y^1)e^{G^{i}(\hat y^i)},
\eeq
where $\phi^{i}(y^1)$ are arbitrary functions of one variable $y^1$.\par
Calculating now derivative by $y^j,j\ne i,1$ and substituting obtained relation into the second equation in 10.9 we see that
\[
f_{i}\ne f_{j}\Rightarrow G^{i}_{,ij}=0.
\]
Restrict now to the case where $f_{i}\ne f_{j}$. Then for all $i,j>1$ such that $i\ne j$, $G^{i}_{,ij}=0$.  This, finally gives us
\beq
h^{0}=\sum_{i>1}\phi^{i}(y^1)e^{G^{i}(y^i)},
\eeq
for arbitrary functions $\phi^i(y^1),G^{i}(y^i)$.\par
This Example illustrates the fact of the following pattern; if entries of constitutive relations depend on some variable in non-symmetrical way (temperature and density are obvious examples), entries of supplementary balance laws will depend on the same variable in non-symmetrical way. In the case of Cattaneo heat propagation (see below, Sec.12) such variable is temperature.  As a result, in the expression of the internal energy and in the entropy balance there appears the new constitutive function $\hat{\lambda}^{0}(\va).$
\end{example}

\vskip1cm

In practical situations it is more convenient to use variables $y^i$ because using $w^i$ one has to express $F^{A}_{i}$ as functions of $w^i$ and that leads to the complex expressions.  That is why we will rewrite the compatibility equation - defining system (61.3) in terms of variables $y^i$. We will use the relations \[ dw^i=F^{0}_{i,y^j}dy^j,\ \frac{\partial h^0}{\partial w^i}=\frac{\partial h^0}{\partial y^k}\frac{\partial y^k}{\partial w^i}.\]
Introduce the matrix $W^{i}_{j}=\frac{\partial F^{0}_{i}}{\partial y^j}$ and \emph{assume that this matrix is non-degenerate} (regular case).  Let $W^{-1}$ be the inverse matrix.  Then we have in (10.1):
\[
F^{A}_{i,j}\frac{\partial h^0}{\partial w^i}dw^j= F^{A}_{i,j}\frac{\partial h^0}{\partial y^k}\frac{\partial y^k}{\partial w^i}F^{0}_{j,y^l}dy^l = F^{A}_{i,j}\frac{\partial h^0}{\partial y^k}W^{-1\ k}_{i}W^{j}_{l}dy^l.
\]
Exterior differential of a function or forms is invariant under the change of variables.  Thus, condition (1o.1) can be written in the form
\beq
d_{v}\left[\left(W^{-1\ k}_{i} F^{A}_{i,j}W^{j}_{l} \right)\frac{\partial h^0}{\partial y^k}dy^l \right]=d_{v}\left[\Psi^{A\ k}_{l}\frac{\partial h^0}{\partial y^k}dy^l \right]=d_{v}(\Psi^{A}_{\mathcal{C}}\cdot d_{v}h^0 ) =0.
\eeq
Here $\Psi^{A}_{\mathcal{C}}$ is the vector of (1,1)-tensor fields in the vertical space.  for each $A=1,2,3$ defines the endomorphism of the vertical cotangent bundle $V^{*}\rightarrow Y$. $\Psi^{A}_{\mathcal{C}}$ is the "constitutive" quantity generalizing tensor of characteristic velocities in the theory of hyperbolic waves. It is defined by the constitutive relation $\mathcal{C}$.

\begin{remark} In these variables linear solutions have the form $h^0 =a_{i}F^{0}_{i}(y)+a,\ a_{i},a \in \mathbb{R}$.
\end{remark}

\begin{example} Consider the system (10.1) for Cattaneo heat propagation balance system (7.6).  We have here $y^0=\va; y^A=q^A,A=1,2,3$. Therefore,
\beq
W=F^{0}_{i,j}=\begin{pmatrix}\e_{,\va} & \e_{,q^1} &\e_{,q^2}& \e_{,q^3}\\
\tau_{,\va}q^1 & \tau &0&0\\
\tau_{,\va}q^2 &0 &\tau & 0\\
\tau_{,\va}q^3 &0&0&\tau
\end{pmatrix}.
\eeq
We have $Det(F^{0}_{i,j})=\tau^2 (\tau \e_{,\va}-\tau_{,\va}q^A \e_{,q^A})$.  Thus, \textbf{the regularity condition for constitutive relation of Cattaneo balance system} is
\beq
\tau^2 (\tau \e_{,\va}-\tau_{,\va}q^A \e_{,q^A})\ne 0.
\eeq
 Calculating the matrices $F^{A}_{i,j}$ we get
\beq
F^{1}_{i,j}=\begin{pmatrix}0 & 1 & 0 & 0\\
\Lambda(\va )_{,\va} & 0 & 0 & 0\\
0 & 0 & 0 & 0\\
0 & 0 & 0 & 0
\end{pmatrix},\ F^{2}_{i,j}=\begin{pmatrix}0 & 0 & 1 & 0\\
0 & 0 & 0 & 0\\
\Lambda(\va )_{,\va} & 0 & 0 & 0\\
0 & 0 & 0 & 0
\end{pmatrix},\ F^{3}_{i,j}=\begin{pmatrix}0 & 0 & 0 & 1\\
0 & 0 & 0 & 0\\
0 & 0 & 0 & 0\\
\Lambda(\va )_{,\va} & 0 & 0 & 0
\end{pmatrix}.
\eeq
It is easy to see that rank of each of these matrices is two, \emph{their kernels do not intersect}: $\bigcap_{A}Ker (F^{A}_{i,j})=\{ 0\}$. As a result, \textbf{defining system for Cattaneo balance system is elliptic}.\par
 More then this, square $(F^{A}_{i,j})^2$ of each of this matrix is equal to $\Lambda_{,\va}$ times the projector to the two-dim plane $<\va ,q^A>$ in the vector space $U=<\va ,q^A,A=1,2,3>$. Derivative $\Lambda_{,\va}$ is the \emph{heat conductivity} coefficient that is nonnegative due to the residual entropy inequality (Muller, TD, Sec.1.3.2).\par
In the basis $e_0 =\va ;e_A =q^A$ eigenvalues of $F^{1}_{i,j}$ are $e_{0}\pm \sqrt{\Lambda_{,\va}}e_A, e_{A+1},e_{A+2}$ (cyclically counted) and as a result, these matrices have no common eigenvectors.  As a result, \textbf{defining system (10.1) for Cattaneo system is holonomic}.
\end{example}

\section{\textbf{Defining system of the 1st order.}}
In this section we rewrite defining system (10.1) as the system of equations of the first order for the generating function $h^0$ and 3 additional functions $h^A,A=1,2,3.$ \par

   Defining system (10.1)
\[
F^{A}_{i,w^j}h^{0}_{,w^i w^k}=F^{A}_{i,w^k}h^{0}_{,w^i w^j}
\]
can be written in the form
\[
\partial_{w^k}(F^{A}_{i,w^j}h^{0}_{,w^i})-F^{A}_{i,w^j w^k}h^{0}_{,w^i}=\partial_{w^j}(F^{A}_{i,w^k}h^{0}_{,w^i})-F^{A}_{i,w^j w^k}h^{0}_{,w^i},
\]
or
\[
\partial_{w^k}(F^{A}_{i,w^j}h^{0}_{,w^i})=\partial_{w^j}(F^{A}_{i,w^k}h^{0}_{,w^i}),\ \forall\ k,j.
\]
This condition is locally equivalent to the original LL-system (6.1) - condition of existing of the potentials $h^A$ such that
\beq
F^{A}_{i,w^j}h^{0}_{,w^i}=\partial_{w^j} h^A,\ A=1,2,3; j=1,\ldots, m.
\eeq
Last condition can be understood as the relation between the vertical differential of the function $h^0$ transformed by the endomorphism of the vertical cotangent bundle $\pi_{V}:V(\pi)^{*}\rightarrow Y$ and the vertical differential of the potential $h^A$:
\beq
\mathcal{F}^{A}d_{v}h^0 =d_{v}h^A,
\eeq
where $\mathcal{F}^{A}$ is the endomorphism of the vertical cotangent bundle defined by
\[\mathcal{F}^{A}(dw^i )=F^{A}_{i,w^k}dw^k=d_{v}F^{A}_{i}.\]
Multiply the $A$-th equality by the real parameter $\epsilon_A$ and add them.  Then, introducing the endomorphism  $\mathcal{F}(\epsilon)=\epsilon_{A}\mathcal{F}^A$ of the bundle $V(\pi)^{*}\rightarrow Y$ and the function $h(x,w,\epsilon)=\epsilon_{A}h^{A}$,
we can write the last equality in the form
\beq
\mathcal{F}(\bar{\epsilon} )d_{v}h^0 =d_{v}h(\epsilon,x,w)\Leftrightarrow \mathcal{F}({\epsilon})_{i,w^j}h^{0}_{,w^i}=h(\epsilon)_{,w^j}.
\eeq
Here $\mathcal{F}({\epsilon})_{i}=\epsilon_{A}F^{A}_{i}.$  Vector $\bar{\epsilon} \in \mathbb{R}^3$ is the parameter. Last relation should be fulfilled for all values of $\bar{\epsilon}$.
\begin{definition}
Balance system ($\bigstar$) will be called C-regular if for some vector $\bar{\epsilon}^* \in R^3$, the matrix $\mathcal{F}({\bar{\epsilon}})_{i,w^j}=\epsilon_{A}F^{A}_{i,w^j}$ is nondegenerate: $det(\mathcal{F}({\bar{\epsilon}^*})_{i,w^j})\ne 0$ (and the corresponding endomorphism is reversible).
\end{definition}
If the condition of this definition is fulfilled for some $\bar{\epsilon}^*$, denote by $\mathcal{F}^{-1}(\bar{\epsilon}^*)$ - the inverse endomorphism of vertical cotangent bundle.  Applying this endomorphism to (11.3) we get the formula for the differential of the density $h^0 (x,w)$.
\beq
d_{v}h^0 =\mathcal{F}^{-1}(\bar{\epsilon}^*)d_{v}h(\bar{\epsilon}^*).
\eeq

Let now we assume that the balance system is C-regular and try to reverse the arguments leading to the formula (11.4).  Namely, let $h^A (w)$ be some functions and let the $m\times m$-matrix $\mathcal{F}(\bar{\epsilon})$ be invertible for $\bar \epsilon$ in an open set $\mathcal{E}\subset R^3$.\par
Then, there are two questions:
\begin{enumerate}
\item Is the expression $\mathcal{F}^{-1}(\bar{\epsilon}^*)d_{v}h({\bar \epsilon}^*)$ the vertical differential of a function $h^0$?
\item Under which condition the function $h^0$ does not depend on the parameter $\bar \epsilon$?
\end{enumerate}
The answer to the first question obviously is: there exists a function $h^0$ such that for ${\bar \epsilon}^*$ equality (11.4) is fulfilled iff
\beq
d_{v}\left(\mathcal{F}^{-1}(\bar{\epsilon}^*) d_{v}h(\bar{\epsilon}^*)\right)=0.
\eeq
In this equality the matrix of 1-forms acts on the 1-form.\par
Assume now that (11.4) is fulfilled for some $\epsilon^*\in E.$  For correctness, the last result should not depend on the choice of $\bar{\epsilon}^*$ and, therefore, equality (11.4) has to be fulfilled in the whole connected component of the set $\mathcal{E}$, containing $\epsilon^*$.\par
This condition is fulfilled at least at the named connected component iff the derivatives by $\epsilon_{A}$ of the right side in (11.4) are identically zero in the whole connected component. This leads to the system of equations
\beq
\frac{\partial}{\partial \epsilon^A}\left(\mathcal{F}^{-1}(\bar{\epsilon})d_{v}h(\bar{\epsilon})\right)=0,
\eeq
for $h^A$ whose solvability conditions delivers the constitutive restrictions on the balance system ($\bigstar$). Calculating derivatives here and using the fact that for a matrix function $F(\epsilon)$,  $\partial_{\epsilon}F^{-1}=-F^{-1}\partial_{\epsilon}F F^{-1},$ we see that the condition

\beq
-\mathcal{F}^{-1}(\bar{\epsilon})[\frac{\partial}{\partial \epsilon^A}\mathcal{F}(\bar{\epsilon} )] \mathcal{F}^{-1}(\bar{\epsilon})d_{v}h(\bar{\epsilon})+\mathcal{F}^{-1}(\bar{\epsilon})\frac{\partial}{\partial \epsilon^A}d_{v}h(\bar{\epsilon})=0,\ \bar \epsilon \in E
\eeq
or, applying $\mathcal{F}(\bar{\epsilon} )$ to both parts and calculating derivative in the last term,
\beq
\mathcal{F}^{A} \mathcal{F}^{-1}(\bar{\epsilon})d_{v}h(\bar{\epsilon})=d_{v}V^A,\ A=1,2,3.
\eeq
has to be fulfilled. More specifically, last condition has the form
\beq
\mathcal{F}^{A} \mathcal{F}^{-1}(\bar{\epsilon})\epsilon_B d_{v}h^B-d_{v}h^A=0,\Leftrightarrow (\mathcal{F}^{A} \mathcal{F}^{-1}(\bar{\epsilon})\epsilon_B-\delta^{A}_{B}\otimes I_{m})d_{v}h^B=0\ A=1,2,3,
\eeq
or, the form
\[
(\mathcal{F}^{A} \epsilon_B-\delta^{A}_{B}\otimes \mathcal{F}(\bar{\epsilon}))\mathcal{F}^{-1}(\bar{\epsilon})d_{v}h^B=0,\ A=1,2,3.
\]
These considerations can be resumed in the form of
\begin{theorem}Let $\bigstar $ be a RET balance system of the RET type and let (3.1) be a supplementary balance law such that corresponding main fields $\lambda^i$ are functionally independent.
\begin{enumerate}
\item Then the density $h^0$ of this SBL satisfies to the system
\beq
F^{A}_{i,w^j}h^{0}_{,w^i}=\partial_{w^j} h^A,\ A=1,2,3; j=1,\ldots, m
\eeq
for some functions $h^A (w^i)$.\par
\item If for some real parameters $\bar{\epsilon} =\{ \epsilon_A,A=1,2,3\}$ the matrix $\mathcal{F}({\bar{\epsilon}})_{i,w^j}=\epsilon_{A}F^{A}_{i,w^j}$ is non-degenerate, then
\beq
d_{v}h^0 =\mathcal{F}^{-1}(\bar{\epsilon}^*)d_{v}h(\bar{\epsilon}^*),
\eeq
where $h(\bar{\epsilon}^*)=\epsilon_{A}h^A$.
\item Functions $h^A(w^i)$ generate the SBL for the balance system ($\bigstar$) if and only if they satisfy to the compatibility system
\beq
\begin{cases}
d_{v}\left(\mathcal{F}^{-1}(\bar{\epsilon}^*) d_{v}h(\bar{\epsilon}^*)\right)=0,\\
(\mathcal{F}^{A} \epsilon_B-\delta^{A}_{B}\otimes \mathcal{F}(\bar{\epsilon}))\mathcal{F}^{-1}(\bar{\epsilon})d_{v}h^B=0,\ A=1,2,3.
\end{cases}
\eeq
\end{enumerate}
\end{theorem}
\begin{remark}
Compatibility conditions of the system of equations for $V^A$ delivers the constitutive conditions on the balance system $\bigstar$.
\end{remark}
\begin{remark}
Equation
\[
d_{v}h^0 -\mathcal{F}^{-1}(\bar{\epsilon}^*)d_{v}h(\bar{\epsilon}^*)=0
\]
can be considered as an abstract infinitesimal version of the Gibbs relation of the locally equilibrium thermodynamics, \cite{GP}.
\end{remark}

\begin{remark}
It is interesting to see if the set $\mathcal{E}$ is connected and if not, do we have the same or different SBL from the same $h^A$?
\end{remark}

\vfill \eject

\vfill \eject
\section{\textbf{Cattaneo heat propagation system.}}
In this section we determine the form of all supplementary balance laws for the Cattaneo heat propagation system (7.6). Instead of solving the LL-equations (4.5) directly, we will use the exterior systems generated by vertical differentials of the flux components $F^{A}_{i}.$
\par
Consider the heat propagation model containing the temperature $\vartheta$ and heat flux $q$ as the independent dynamical fields. $y^0 =\vartheta, y^A=q^A,\ A=1,2,3$. \par

Balance equations of this model have the form

\beq
\begin{cases}
\partial_{t}(\rho \epsilon)+div(q)=0,\\
\partial_{t}(\tau q)+ \nabla \Lambda(\vartheta) =-q.
\end{cases}
\eeq
Second equation can be rewritten in the conventional form
\[\partial_{t}(\tau q)+ \lambda \cdot\nabla \vartheta =-q,\]
where $\lambda =\frac{\partial \Lambda }{\partial \vartheta}.$  If coefficient $\lambda$ may depend on the density $\rho$, equation is more complex. \par
Since $\rho$ is not considered here as a dynamical variable, we merge it with the field $\epsilon$ and from now on and till the end it will be omitted. On the other hand, in this model the the energy $\epsilon$ depends on temperature $\vartheta$ \emph{and on the heat flux} $q$ (see \cite{JCL},Sec.2.1.2) or, by change of variables, temperature $\vartheta=\vartheta(\epsilon, q)$ will be considered as the function of dynamical variables.\par
  All variables may depend on $x^\mu$ directly or though $\rho$.\par
Cattaneo equation $q+\tau \partial_{t}(q)=-\lambda\cdot \nabla \vartheta$ has the form of the vectorial balance law and, as a result there is no need for the constitutive relations to depend on the derivatives of the basic fields.  No derivatives appears in the constitutive relation, therefore, this is the RET model. In the second equation there is a nonzero production $\Pi^A=-q^A$. Model is homogeneous, there is no explicit dependence of any functions on $t,x^A$.\par

Constitutive relation specify dependence of the internal energy $\epsilon$ on $\vartheta ,q$ and possible dependence of coefficients $\tau, \lambda$ on the temperature.  Simplest case is the linear relation $\epsilon =k\vartheta$, but for our purposes it is too restrictive, see \cite{JCL}, Sec.2.1. \par

\textbf{Vertical variables} are here $\vartheta, q^A$.\par
For the components of constitutive relation we have
\beq
F^{0}_{0}=\epsilon (\vartheta,q^A),\ F^{0}_{A}=\tau(\vartheta) q^A,\ F^{B}_{0}=q^B,\ F^{B}_{C}=\delta^{B}_{C}\Lambda (\vartheta),\ A,B,C=1,2,3.
\eeq
\par
We start with the $i\times \mu$ matrix

\[
F^{\mu}_{i}=\begin{pmatrix} \epsilon & \tau q^1 & \tau q^2 & \tau q^3\\
 q^1 & \Lambda(\vartheta) & 0 & 0\\
 q^2 &0 & \Lambda(\vartheta) & 0\\
 q^3 & 0 & 0 & \Lambda(\vartheta)
\end{pmatrix}.
\]
Assuming that coefficients $\tau$ and the function $\Lambda$ are independent on the vertical variables $ q^A$ we get the vertical differentials
\[
d_v F^{\mu}_{i}=\begin{pmatrix} \e_{\va}d\va +\e_{q^A}dq^A  & \tau_{\va}q^1 d\va+ \tau dq^1 &\tau_{\va}q^2 d\va+\tau dq^2 &\tau_{\va}q^2 d\va+\tau dq^3\\
 dq^1 & \Lambda_{\va} d\vartheta & 0 & 0\\
 dq^2 &0 & \Lambda_{\va} d\vartheta & 0\\
 dq^3 & 0 & 0 &\Lambda_{\va} d\vartheta
\end{pmatrix}.
\]
\par
Left side of LL-subsystem for $\mu =0$ is
\begin{multline}
\mathcal{H}^0 (\bar{\lambda })=(\lambda^0,\ \lambda^1,\ \lambda^2,\ \lambda^3)\cdot \begin{pmatrix} \e_{\va} & \e_{q^1} & \e_{q^2} & \e_{q^3}\\
\tau_{\va}q^1 & \tau & 0 & 0\\
\tau_{\va}q^2 & 0 & \tau & 0 \\
\tau_{\va}q^3 & 0 & 0 & \tau
\end{pmatrix} =\begin{pmatrix} \partial_{\alpha^0}K^0 \\ \partial_{\alpha^1}K^0 \\ \partial_{\alpha^2}K^0 \\ \partial_{\alpha^3}K^0 \end{pmatrix}.
\end{multline}
As a result, LL subsystem for $\mu =0$ takes the form

\beq
\begin{cases} \lambda^0 \e_{\va}+\tau_{\va}\lambda^A q^A  =K^0_{,\va},\\ \ \lambda^0 \e_{q^A}+\lambda^A \tau  =K^{0}_{q^A}\end{cases},\ A=1,2,3 .
\eeq

For $\mu=1$ we have in the same way

\beq
(\lambda^0,\ \lambda^1,\ \lambda^2,\ \lambda^3)\cdot \begin{pmatrix} 0 & 1 & 0 & 0\\
\Lambda_{\va} & 0 & 0 & 0\\
0 & 0 & 0& 0 \\
0 & 0 & 0 & 0
\end{pmatrix} =\begin{pmatrix} K^{1}_{,\va} \\ K^{1}_{,q^2} \\ K^{1}_{,q^1} \\ K^{1}_{,q^3} \end{pmatrix} \Leftrightarrow \ \begin{cases} \Lambda_{\va} \lambda^1  =K^1_{,\va},\\ \ \lambda^0  =K^{1}_{q^1},\\ 0=K^{1}_{q^2},\\ 0=K^{1}_{q^3} .\end{cases}
\eeq

Repeating this for $A=2,3$ we get the subsystems of LL-system for $A=1,2,3$  in the form
\beq
\begin{cases} \Lambda_{\va} \lambda^A  =K^A_{,\va},\\ \ \lambda^0  =K^{A}_{q^A},\\ 0=K^{A}_{q^{A+1}},\\ 0=K^{A}_{q^{A+2}},\ A=1,2,3 .\end{cases}
\eeq

Looking at systems (12.4-6) we see that if we make the change of variables: $\tva =\Lambda (\va )$ then the systems equations we get takes the form ( wherever is the derivative by $\va$ we multiply this equation by $\Lambda_{,\va}$)

\beq
\begin{cases} \lambda^0 \e_{\tva}+\tau_{\tva}\lambda^A q^A  =K^0_{,\tva},\\ \ \lambda^0 \e_{q^A}+\lambda^A \tau  =K^{0}_{q^A}\end{cases};\
\begin{cases}    K^A_{,\tva}=\lambda^A,\\ \ K^{A}_{q^B}=\lambda^0 \delta^{A}_{B}\end{cases},\ A,B=1,2,3 .
\eeq
Second subsystem is equivalent to the relation
\[
d_{v}K^A=\lambda^A d\vartheta +\lambda^0 dq^A.
\]

Integrability conditions of these form imply $K^A =K^A (\tva ,q^A)$ and
\[
K^{A}_{q^A}=\lambda^{0},\ A=1,2,3\ \Rightarrow \lambda^0 =\lambda^0 (\tva).
\]
Integrating equation $K^{A}_{q^A}=\lambda^{0}(\tva)$ by $q^Q$ we get
\beq
K^{A}=\lambda^{0}(\tva)q^A +\tilde{K}^{A}(\tva).
\eeq

First equation of each system now takes the form
\beq
\lambda^A =K^{A}_{\tva}=\lambda^{0}_{\tva}q^A +\tilde{K}^{A}_{,\tva}(\tva).
\eeq
Substituting these expressions for $\lambda^A$ into the 0-th system
\[
 \begin{cases} \lambda^0 \e_{\tva}+\tau_{\tva}\lambda^A q^A  =K^0_{,\tva},\\ \ \lambda^0 \e_{q^A}+\lambda^A \tau  =K^{0}_{q^A}\end{cases},\ A=1,2,3,
\]
we get
\beq
 \begin{cases} K^0_{,\tva}=\lambda^0 \e_{\tva}+\tau_{\tva}(\lambda^{0}_{\tva}\Vert q\Vert^2 +\tilde{K}^{A}_{,\tva}(\tva)q^A)  ,\\ \ K^{0}_{q^A}=\lambda^0 \e_{q^A}+\tau (\lambda^{0}_{\tva}q^A +\tilde{K}^{A}_{,\tva}(\tva)) \end{cases},\ A=1,2,3,
\eeq
where $\Vert q\Vert^2 =\sum_{A}q^{A\ 2}$.\par

Integrating $A$-th equation by $q^A$ and comparing results for different $A$ we obtain the following representation
\beq
K^0 =\lambda^0 \e +\tau(\tva )[\frac{1}{2}\lambda^{0}_{\tva}\Vert q\Vert^2+\tilde{K}^{A}_{,\tva}(\tva)q^A]+f(\tva)
\eeq
for some function $f(\tva )$.\par

Calculate derivative by $\tva$ in the last formula for $K^0$ and subtract the first formula of the previous system.  We get
\beq
0=\lambda^{0}_{,\tva} \e +\tau(\tva )[(\frac{1}{2}\lambda^{0}_{\tva}\Vert q\Vert^2+\tilde{K}^{A}_{,\tva}(\tva)q^A)]_{,\tva}-\frac{1}{2}\tau_{,\tva}\lambda^{0}_{,\tva}\Vert q\Vert^2  +f_{,\tva}(\tva).
\eeq
This is the compatibility condition for the system (12.11) for $K^0$.  As such, it is realization of the general compatibility system (6.2).\par

Take $q^A=0$ in the last equation, i.e. consider the case where \emph{there are no heat flux}.  Then the internal energy reduces to its \emph{equilibrium value } $\e^{eq}(\tva)$ and we get $f_{,\tva}(\tva)=-\lambda^{0}_{,\tva} \e^{eq}$.  Integrating here we find
\beq
f(\tva )=f_{0}-\int^{\tva}\lambda^{0}_{,\tva}(s) \e^{eq}(s) ds.
\eeq
Substituting this value for $f$ into the previous formula and we get expressions for $K^\mu$:

\beq
\begin{cases}
K^0 =\lambda^0 \e -\int^{\tva}\lambda^{0}_{,\tva} \e^{eq} ds +\tau(\tva )[\frac{1}{2}\lambda^{0}_{\tva}\Vert q\Vert^2+\tilde{K}^{A}_{,\tva}(\tva)q^A]+f_{0},\\
K^{A}=\lambda^{0}(\tva)q^A +\tilde{K}^{A}(\tva),\ A=1,2,3.
\end{cases}
\eeq
In addition to this, from (12.12) and obtained expression for $f(\tva)$, we get the \emph{expression for internal energy}
\beq
\e=\e^{eq}(\tva)+\frac{1}{2}\tau_{,\tva }\Vert q\Vert^2-\frac{\tau (\tva)}{\lambda^{0}_{\tva}(\tva)}\left[ \frac{1}{2}\lambda^{0}_{,\tva \tva}\Vert q\Vert^2+\tilde{K}^{A}_{,\tva\tva}(\tva)q^A\right].
\eeq
\textbf{This expression present the restriction to the constitutive relations in Cattaneo model placed on it by the entropy principle.}\par

Zero-th main field $\lambda^0$ is an arbitrary function of $\tva $ while $\lambda^A$ is given by (12.10):
\beq
\lambda^A = (\lambda^{0}_{\tva}q^A +\tilde{K}^{A}_{,\tva}(\tva)).
\eeq
Using this we find the source term
\beq
Q=\lambda^A \Pi_{A}=-\lambda^A q^A =-(\lambda^{0}_{\tva}\Vert q\Vert^2 +\tilde{K}^{A}_{,\tva}(\tva)q^A).
\eeq

Now we combine obtained expressions for components of a secondary balance law.  We have to take into account that (see (12.7)) that the LL-system of relations defines $K^\mu $ only $  mod \ C^{\infty}(X)$ (RET case!).  This means first of all that all the functions may depend explicitly on $x^\mu$. For energy $\e$, field $\Lambda(\va )$ and the coefficient $\tau$ this dependence is determined by constitutive relations and is, therefore, fixed.
Looking at (12.15) we see that the coefficients f linear and quadratic by $q^A$ are also defined by the constitutive relation, i.e. in the representation
\beq
\e=\e^{eq}(\tva)+\mu(\tva )\Vert q\Vert^2+ M_{A}(\tva )q^A = \e^{eq}(\tva)+\frac{1}{2}\tau_{,\tva }\Vert q\Vert^2-\frac{\tau (\tva)}{\lambda^{0}_{\tva}(\tva)}\left[ \frac{1}{2}\lambda^{0}_{,\tva \tva}\Vert q\Vert^2+\tilde{K}^{A}_{,\tva\tva}(\tva)q^A\right],
\eeq
coefficients
\beq
\mu(\tva ,x)=\frac{1}{2}\tau_{,\tva }-\frac{1}{2}\frac{\tau (\tva)}{\lambda^{0}_{\tva}(\tva)}\lambda^{0}_{\tva\tva},\ M_{A}=-\frac{\tau (\tva)}{\lambda^{0}_{\tva}(\tva)}\tilde{K}^{A}_{,\tva\tva}(\tva)
\eeq
are defined by the CR - by expression of internal energy as the quadratic function of the heat flux.\par
Rewriting the first relation we get
\begin{multline}
\left(ln(\lambda^{0}_{\tva})\right)_{,\tva}=ln(\tau )_{,\tva} -2\frac{\mu(\tva )}{\tau (\tva)}  \Rightarrow ln(\lambda^{0}_{\tva})=ln(\tau)+b^0-2\int^{\tva}\frac{\mu}{\tau}(s)ds \Rightarrow \\ \Rightarrow \lambda^{0}_{\tva}=\alpha \tau e^{-2\int^{\tva}\frac{\mu}{\tau}(s)ds},\ \alpha=e^{b^0}>0.
\end{multline}
From this relation we find
\beq
 \lambda^{0}(\tva,x)=a^0+\alpha\hat{\lambda}^{0}=a^0+\alpha \int^{\tva}[\tau e^{-2\int^{u}\frac{\mu (s)}{\tau (s)}ds}]du
\eeq
Using this in the second formula (12.19) we get the expression for coefficients $\tilde{K}^{A}$
\begin{multline}
\tilde{K}^{A}_{,\tva\tva}=-M_{A}\cdot \frac{\lambda^{0}_{\tva}(\tva)}{\tau (\tva)}=-M_{A}\alpha e^{-2\int^{\tva}\frac{\mu}{\tau}(s)ds}\Rightarrow \\ \Rightarrow \tilde{K}^{A} =k^A \tva +m^A +\alpha \cdot \hat{K}^{A}(\tva ) =k^A \tva +m^A-\alpha \int^{\tva} dw\int^{w}[M_{A}(u)e^{-2\int^{u}\frac{\mu}{\tau}(s)ds}]du.
\end{multline}

Functions $\hat{K}^{A}(\tva )$ are defined by the second formula in the second line.\par
Thus, functions $\lambda^{0}_{\va},\tilde{K}^{A}_{,\va\va}$ are defined by the constitutive relations while coefficients $\alpha >0,a^0,k^A,m^A$ are arbitrary functions of $x^\mu$.
\begin{remark} It would be interesting to chose lower limit in the integrals in previous formulas from physical reasons.  Then the coefficients $a^0,k^A,m^A$ might have some physical meaning too.
\end{remark}
Combine obtained results and returning to the variable $\va$ we get (here we repeatedly use the relation $f_{,\tva}=\va_{,\tva} f_{,\va}=(\tva_{,\va})^{-1} f_{,\va}=\Lambda^{-1}_{,\va}f_{,\va}$)
\beq
\begin{cases}
K^0 =\lambda^0 \e -\int^{\tva}\lambda^{0}_{,\tva} \e^{eq} ds +\tau(\tva )[\frac{1}{2}\lambda^{0}_{\tva}\Vert q\Vert^2+\tilde{K}^{A}_{,\tva}(\tva)q^A]+f_{0}=\\=(a^0+\alpha\hat{\lambda}^0 ) \e -\alpha \int^{\va}\hat{\lambda}^{0}_{,\va} \e^{eq} ds +\frac{\tau(\va )}{\Lambda_{,\va}}[\frac{\alpha}{2}\hat{\lambda}^{0}_{\va}\Vert q\Vert^2+(\Lambda_{,\va} k^A+\alpha\hat{K}^{A}_{,\va}(\va))q^A]+f_{0},\\
K^{A}=\lambda^{0}(\tva)q^A +\tilde{K}^{A}(\tva)=(a^0+\alpha\hat{\lambda}^{0}(\va))q^A + k^A \Lambda(\va) +m^A +\alpha\hat{K}^{A}(\va ) ,\ A=1,2,3.\\
Q=-\lambda^A q^A =-(\lambda^{0}_{\tva}\Vert q\Vert^2 +\tilde{K}^{A}_{,\tva}(\tva)q^A)=-\Lambda^{-1}_{,\va}(\lambda^{0}_{\va}\Vert q\Vert^2 +\Lambda_{,\va } k^A q^A +\alpha\hat{K}^{A}_{,\va}(\va)q^A )=\\
-\Lambda^{-1}_{,\va}(\alpha \hat{\lambda}^{0}_{\va}\Vert q\Vert^2 +\Lambda_{,\va } k^A q^A +\alpha\hat{K}^{A}_{,\va}(\va)q^A ).
\end{cases}
\eeq

Collecting obtained results together we present obtained expressions for secondary balance laws first in short form and then - in the form where relations (12.19-20) were used to separate the original balance laws from the general form
\begin{multline}
\begin{pmatrix}K^0\\K^1\\ K^2\\ K^3 \\ Q\end{pmatrix} = \begin{pmatrix}\lambda^0 \e -\int^{\va}\lambda^{0}_{,\va} \e^{eq} ds +\tau(\va )\Lambda^{-1}_{\va}[\frac{1}{2}\lambda^{0}_{\va}\Vert q\Vert^2+\tilde{K}^{A}_{,\va}(\va)q^A]+f_{0}\\ \lambda^{0}(\va)q^1 +\tilde{K}^{1}(\va)\\ \lambda^{0}(\va)q^2 +\tilde{K}^{2}(\va)\\ \lambda^{0}(\va)q^3 +\tilde{K}^{3}(\va)\\ -\Lambda^{-1}_{,\va }(\lambda^{0}_{,\va}\Vert q\Vert^2 +\tilde{K}^{A}_{,\va}(\va)q^A)\end{pmatrix}=\\=a^0 \begin{pmatrix} \e \\ q^1\\q^2\\q^3 \\0
\end{pmatrix}+ \sum_{A}k^A \begin{pmatrix} \tau(\va )q^A \\ \delta^{1}_{A} \Lambda( \va) \\ \delta^{2}_{A} \Lambda(\va) \\ \delta^{3}_{A} \Lambda(\va) \\ - q^A\end{pmatrix}
+\begin{pmatrix} \hat{K}^{A}_{,\va}(\va)q^A\\ \hat{K}^{1}(\va)\\ \hat{K}^{2}(\va)\\ +\hat{K}^{3}(\va) \\ -\Lambda^{-1}_{,\va }\hat{K}^{A}_{,\va}(\va)q^A
\end{pmatrix}
+ \alpha \begin{pmatrix}\hat{\lambda}^0 \e -\int^{\va}\lambda^{0}_{,\va} \e^{eq} ds +\tau(\va )\Lambda^{-1}_{\va}[\frac{1}{2}\lambda^{0}_{,\va}\Vert q\Vert^2]\\ \hat{\lambda}^{0}(\va)q^1 \\ \hat{\lambda}^{0}(\va)q^2 \\ \hat{\lambda}^{0}(\va)q^3 \\ -\Lambda^{-1}_{,\va }\hat{\lambda}^{0}_{,\va}\Vert q\Vert^2 \end{pmatrix}+\begin{pmatrix} f_0\\ m^1\\m^2\\m^3 \\0
\end{pmatrix}.
\end{multline}
\begin{remark} Notice the correspondence between the tensor structure of the basic fields of Cattaneo system - one scalar field (temperature $\va$) and one vector field (heat flux $q^A, A=1,2,3$) and the structure of space $\mathcal{SBL}(C)$ of supplementary balance laws - elements of $\mathcal{SBL}(C)$ depend on one scalar function of temperature $\lambda^{0}(\va)$ and one covector function of temperature $\hat{K}_{A}$.
\end{remark}

Returning to the variable $\va$ in the expression (12.15) and using the relation $\partial_{\tva}=\frac{1}{\Lambda (\va)_{,\va}} \partial_{\va}$  we get the expression for the internal energy
\begin{multline}
\e =\e^{eq}(\va)+\frac{\tau_{,\va}}{2\Lambda_{,\va}}\Vert q\Vert^2 -\frac{\tau(\va)}{\lambda^{0}_{,\va}}\left[ \frac{1}{2}\left(\frac{\lambda^{0}_{,\va}}{\Lambda_{,\va}} \right)_{,\va}\Vert q\Vert^2+ \left(\frac{{\tilde K}^{A}_{,\va}}{\Lambda_{,\va}}\right)_{,\va}q^A \right]=\\ = ^{\Lambda_{,\va}=\kappa-const}\e^{eq}(\va) +\frac{\tau_{,\va}}{2\kappa }\Vert q\Vert^2 -\frac{\tau(\va)}{\kappa\lambda^{0}_{,\va}}\left[ \frac{1}{2}\lambda^{0}_{,\va\va}\Vert q\Vert^2+ {\tilde K}^{A}_{,\va\va}q^A \right].
\end{multline}
Notice that for $\lambda^{0}=0$, balance law given by the 4th column in (12.24) vanish. The same is true for deformations of the Cattaneo equation defined by the third column.\par

First and second balance laws in the system (12.24) are the balance laws of the original Cattaneo system. Last one is the trivial balance law (see Sec.3).  Third column gives the balance law
\begin{multline}
\partial_{t}\left[ \hat{\lambda}^0 \e -\int^{\va}\lambda^{0}_{,\va} \e^{eq}ds +\tau(\va )\Lambda^{-1}_{\va}[\frac{1}{2}\lambda^{0}_{\va}\Vert q\Vert^2+\hat{K}^{A}_{,\va}(\va)q^A]\right] +\partial_{x^A}\left[\hat{\lambda}^{0}(\va)q^A +\hat{K}^{A}(\va) \right]=\\ = -\Lambda^{-1}_{,\va }(\hat{\lambda}^{0}_{\va}\Vert q\Vert^2 +\hat{K}^{A}_{,\va}(\va)q^A).
\end{multline}

Source/production term in this equation has the form

\begin{multline}
 -\Lambda^{-1}_{,\va }(\hat{\lambda}^{0}_{\va}\Vert q\Vert^2 +\hat{K}^{A}_{,\va}(\va)q^A)=-\Lambda^{-1}_{,\va }\hat{\lambda}^{0}_{\va}(\Vert q\Vert^2 +\frac{\hat{K}^{A}_{,\va}(\va)}{\hat{\lambda}^{0}_{\va}}q^A)=\\=
 -\Lambda^{-1}_{,\va }\hat{\lambda}^{0}_{\va}\left[\sum_{A}( q^A +\frac{\hat{K}^{A}_{,\va}(\va)}{2\hat{\lambda}^{0}_{\va}})^2 -\sum_{A} \left( \frac{\hat{K}^{A}_{,\va}(\va)}{2\hat{\lambda}^{0}_{\va}}\right)^2\right]
\end{multline}

For a fixed $\va$ this expression may have constant sign as the function of $q^A$ if and only if all $\hat{K}^{A}_{,\va}(\va)=0$. Therefore this is possible only if
the internal energy (12.25) has the form
\beq
\e=\e^{eq}(\va)+\left[ \frac{\tau_{,\va}}{2\Lambda_{,\va}} -\frac{\tau(\va)}{2\hat{\lambda}^{0}_{,\va}} \left(\frac{\hat{\lambda}^{0}_{,\va}}{\Lambda_{,\va}} \right)_{,\va} \right] \Vert q\Vert^2
\eeq
with some function ${\hat \lambda}^{0}(\va)$, Cattaneo model has the supplementary balance law
\begin{multline}
\partial_{t}\left[ \hat{\lambda}^0 \e -\int^{\va}\hat{\lambda}^{0}_{,\va} \e^{eq}ds +\frac{1}{2}\tau(\va )\Lambda^{-1}_{\va}\hat{\lambda}^{0}_{\va}\Vert q\Vert^2\right] +\partial_{x^A}\left[\hat{\lambda}^{0}(\va)q^A \right]= -\Lambda^{-1}_{,\va }\hat{\lambda}^{0}_{\va}\Vert q\Vert^2
\end{multline}
 \textbf{with the production term that may have constant sign - nonnegative, provided}
 \beq
 \Lambda^{-1}_{,\va }\hat{\lambda}^{0}_{\va}\leqq 0.
 \eeq
\textbf{This inequality is the II law of thermodynamics for Cattaneo heat propagation model.}\par
We collect obtained results in the following
\begin{theorem}
\begin{enumerate}
\item For the Cattaneo hear propagation balance system (12.1) compatible with the entropy principle the internal energy has the form (12.25).
All supplementary balance laws are listed in (12.24).  New supplementary balance laws depend on the 4 functions of temperature - $\hat{\lambda}^{0}(\va),\tilde{K}^{A}(\va), A=1,2,3$.
\item Additional balance law (12.26) given by the sum of third and forth columns in (12.24) has the nonnegative production term if and only if the internal energy $\epsilon$ is given by the formula (12.28) and, in addition, the condition (12.30) holds.
\end{enumerate}
\end{theorem}

\section{\textbf{Conclusion.}}
In this work we study the mathematical aspects of the development in the continuum thermodynamics that started with the pioneering works of B.Coleman, W.Noll and I. Muller in 60th of XX cent. and got its further development mostly in the works of G. Boillat, I-Shis Liu and T.Ruggeri.  This development is known as the "Entropy Principle".  It combined in itself the structural requirement on the form of balance laws of the thermodynamical system (denote it $(\mathcal{C})$) and on the entropy balance law with the convexity condition of the entropy density.\par
First of these requirements has pure mathematical form defining so called "supplementary balance laws" (SBL) associated with the original balance system. Space of SBL can be considered as a kind of natural "closure" of the original balance system.  This vector space includes: original balance laws, the entropy balance, the balance laws corresponding to the symmetries of the balance system and some other balance equations (see, for instance, Example 1). Thus, the space of supplementary balance laws carries important information about the thermodynamical system $(\mathcal{C})$ and deserves the attention.\par
In this text we study the case of Rational Extended Thermodynamics where densities, fluxes and production terms of the balance system do not depend on the derivatives of physical fields $y^i$. We revisit the formalism of "main fields", Lagrange-Liu equations and dual formulation of the balance system in terms more formal then it is done in physical literature. Our main goal here was to obtain and start studying the defining system of equations for the density $h^0$ of a supplementary balance law.  This overdetermined linear system of PDE of second order contains in itself  equations determining all the densities $h^0$ and with them, due to the formalism of RET, the fluxes and sources of SBL. In addition to this, this system contains, in the form of solvability conditions, the constitutive restrictions on the balance equations of the original balance system.\par
We illustrate our results by some simple examples of balance system and by describing all the supplementary balance laws and the constitutive restrictions for the Cattaneo heat propagation system. \par
Study of the space of supplementary balance laws for other physical systems (5 fields fluid, Elasticity, 13 fields) is in the process and will be presented elsewhere.

\section{\textbf{Appendix.}}
\subsection{{Appendix I. Solvability of systems of PDE.}}
Here we remind the following definition (\cite{O}, (2.70)).
\begin{definition}
\begin{enumerate}
 \item
 A system of differential equations $\Delta (x^\mu,y^i (x),z^{i}_{\Lambda }\vert \Lambda \leqq k )=0,$ is called \textbf{locally solvable at a point} $z_{0}=(x_{0},y_{0},\ldots z^{i}_{\Lambda \ 0})\in J^{k}(\pi)$ if there exists a smooth solution $y=s(x)$ defined in a neighborhood of the point $x_0$ which has prescribed "initial condition" $j^{k}s(x_0 )=z_{0}$.
  \item A system called non-degenerate  at a point $z_{0}\in J^{1}_{p}(\pi)$ if it is locally solvable at $z_{o}$ and \textbf{is of maximal rank} at this point.
 \item A system is called locally solvable (non-degenerate) in a domain $D\subset X$ if it is locally solvable (non-degenerate) at each point of domain $D$.
 \end{enumerate}
\end{definition}
\subsection{{Appendix II. Symmetric hyperbolic systems.}}
\begin{definition} A system of linear equations
\beq
A^{\mu}\frac{\partial}{\partial x^\mu}\bar{y}=\overline{\Pi}
\eeq
for a vector function $\bar y (x^\mu)$ is called symmetrical hyperbolic system if
\begin{enumerate}
\item  Matrices $A^\mu(x,y)$ are symmetrical,
\item Matrix $A^{0}$ is positive definite.
\end{enumerate}
\end{definition}

\subsection{{Appendix III. Cartan Lemma.}}
\begin{proposition}(Cartan Lemma,\cite{Ster}., Thm.4.4) Let $V$ be a vector space over a field $k$. Let vectors $x_{i},i=1,\ldots, p$  are linearly independent and for a vectors $y_i\in V, i=1,\ldots, p$,
\[
\sum_{i}x^i \wedge y_i=0.
\]
Then
\[
y_i=\sum_{j=1}^{p}A_{ij}x_j,\ \text{where}\ A_{ij}=A_{ji}.
\]
\end{proposition}
\vskip1cm

\textbf{Acknowledgements}. I would like to express my deep gratitude to Professor T.Ruggeri.  His lecture in Messina at the Thermoconn 2005 was the starting point of my work on this subject and the talks with him at the Montecatini Terme on October 2009 were especially strong stimulus for the study of supplementary balance laws for the balance systems of Continuum Thermodynamics. I am also profoundly thankful to Professor W.Muschik for the useful and deep discussion of the important thermodynamical topics.

\end{document}